%% file: Proxy_voting_Arxiv.tex
\pdfoutput=1

\documentclass[letterpaper]{article}
\usepackage{times}
\usepackage{helvet}
\usepackage{courier}
\usepackage{amsmath}
\usepackage{amssymb}
	\usepackage{amsthm}
\usepackage{mathtools}
\usepackage{blindtext}
\usepackage{graphicx}
\graphicspath{ {images/} }
\usepackage{verbatim}
\usepackage{xcolor}
\usepackage{stmaryrd}
\usepackage{tikz}
\input{defs}
\usepackage{url}
\usepackage{gensymb}
\usepackage{textcomp}

\newcount\Comments  
\definecolor{darkgreen}{rgb}{0,0.6,0}
\newcommand{\kibitz}[2]{\ifnum\Comments=1{\color{#1}{#2}}\fi}
\newcommand{\rmr}[1]{\kibitz{blue}{[RESHEF:#1]}}

\def\shortcite{\cite}

\newcommand{\cut}[1]{#1}

\def\var{\mathbb{V}}
\def\newpar{\vspace{-2mm}\paragraph}
\def\newsubsec{\vspace{-1mm}\subsection}

\begin{document}

\title{Proxy Voting for Better Outcomes}
\author{Gal Cohensius, Shie Manor, Reshef Meir, Eli Meirom and Ariel Orda\\
	Technion---Israel Institute of Technology}
\maketitle

\begin{abstract}
	We consider a social choice problem where only a small number of people out of a large population are sufficiently available or motivated to vote. A common solution to increase participation is to allow voters use a proxy, that is, transfer their voting rights to another voter. Considering social choice problems on metric spaces, we compare voting with and without the use of proxies to see which mechanism better approximates the optimal outcome, and characterize the regimes in which proxy voting is beneficial. 
	
	When voters' opinions are located on an interval, both the median mechanism and the mean mechanism are substantially improved by proxy voting. When voters vote on many binary issues, proxy voting is better when the sample of active voters is too small to provide a good outcome.  Our theoretical results extend to situations where available voters choose strategically whether to participate. We support our theoretical findings with empirical results showing substantial benefits of proxy voting on simulated and real preference data. 
\end{abstract}







\section{Introduction}
\label{sec:intro}
In his 1969 paper, James Miller envisioned a world where technology enables people to vote from their homes~\cite{miller1969program}. With the rise of participatory democracies, the formation of many overlapping online communities, and the increasing use of polls by companies and service providers, this vision is turning into  reality. 

New online voting apps provide an easy way for people to report and aggregate their preferences, from simple direct polls (such as those used by Facebook and Doodle), through encrypted large-scale applications (e.g. \url{electionbuddy.com}), to  sophisticated tools that use AI to guide group selection, such as \url{robovote.org}. 
As a result, each of us is prompted to vote in various formats multiple times a day: we vote for our union members and approve their decisions, on meeting times, and even on the temperature in our office.\footnote{\url{http://design-milk.com/comfy-app-}  \url{brings-group-voting-workplace-thermostat}.}

Is direct democracy coming back? Can it replace representative democracy and parliaments? As it turns out, many online voting instances and polls have low participation rates~\cite{christian2005voting,jonsson2011user}, presumably since most people consider them  insignificant, low-priority, or simply a burden.  
The actual decisions in many of these polls are often taken by a small group of dedicated and active voters, with little or no involvement from most people who could have voted. The outcome in such cases may be completely unrepresentative for the entire population, e.g. if the motivation of the active voters depends on their position or other factors. Even if the set of active voters is selected at random and is thus representative in expectation, there may be too few voters for a reliable outcome. For example, Mueller et al.~\shortcite{mueller1972representative} argue that to function well, such a ``random democracy'' would require over 1000 representatives.

\medskip

\emph{Proxy voting} lets voters who are unable or uninterested to vote themselves transfer their voting rights to another person---a proxy. Proxy voting is common in politics and in corporates~\cite{riddick1985riddick}, and plays an important role in existing and planned systems for e-voting and participatory democracies~\cite{petrik2009participation}. Yet there is only a handful of theoretical models dealing with proxy voting, and our understanding of its effects are limited (see Discussion). 

In this paper, we model voters' positions as points in a metric space aggregated by some function $\g$ (specifically,  Median, Mean, or Majority).  For example, a voter's position may be her preferred pension policy in the union's negotiation with management (say, how much to save on a scale of 0 to 10). The optimal policy is an aggregate over the preferences of all employees. Since actively participating in union's meeting costs time and effort, we consider a subset of \emph{active voters} selected from the population (either at random or by strategic self-selection), and ask whether the accuracy can be improved by allowing inactive voters to use a proxy at no cost. Following Tullock~\shortcite{tullock1967proportional}, we weigh the few active voters (who are used as proxies) according to their number of followers, and assume that inactive voters select the ``nearest'' active voter as a proxy. For example, a person who is unable to attend the next union meeting could use an online app to select a colleague with similar preferences as her proxy, thereby increasing his weight and influencing the outcome in her direction.

The intuition for why proxy voting should increase accuracy is straight-forward: opinions that are more ``central'' or ``representative'' would attract followers and gain weight, whereas the weight of ``outliers'' that distort the outcome will be demoted. However as we will see, this reasoning does not always work in practice.
Thus it is important to understand the conditions in which proxy voting is expected to improve accuracy, especially when voters behave strategically.

%
%

\newsubsec{Contribution and Structure}
We dedicate one section to each common mechanism, and show via  theorems and empirical results that proxy voting usually has a significant positive effect on accuracy, and hence welfare.  For the Median mechanism on a line (Section~\ref{sec:median}), proxy voting may only increase the accuracy, often substantially. For the Mean mechanism on a line (Section~\ref{sec:mean}), we show improvement in expectation if active voters are sampled from the population at random. The last domain contains multiple independent binary issues, where a Majority vote is applied to each issue (Section~\ref{sec:binary}). Here we show that proxy voting essentially leads to a ``dictatorship of the best expert,'' which increases accuracy when the sample is small and/or when voters have high disagreements. Interestingly, results on real preference data are even more positive, and we analyze the reasons in the text.
We further characterize equilibria outcomes when voters strategically choose whether to become active (i.e., use as proxies), and show that most of our results extend this strategic setting.  Results are summarized in Table~\ref{tab:results}.

\section{Preliminaries}

%
%

$\cal X$  is the \emph{space}, or set of possible voter's preferences, or types. In this paper  $\cal X\subseteq \mathbb R^k$ for some $k\geq 1$ dimensions, thus each type can be thought of as a position in space. We use the $\ell_\rho$ distance metric on $\cal X$. 
In particular, we will consider two spaces: an unknown interval $\cal X=[a,b]$ for some $a,b\in {\mathbb R}  \cup \{\pm \infty\}$, and multiple binary issues $\calX = \{0,1\}^k$. Note that this means that all $\ell_\rho$ norms coincide (not true e.g. for $\calX=\mathbb R^2$).

We assume an infinite population of voters, that is given by a distribution $f$ over $\calX$. We say that $f$ over the interval $[a,b]$ is \emph{symmetric} if there is a point $c$ s.t. $f(c-x)=f(c+x)$ for all $x$.  We say that $f$ over the interval $[a,b]$ is \emph{[weakly] single-peaked}  if there is a point $c\in \calX$ s.t. $f$ is [weakly] increasing in $[\min a,z]$ and [weakly] decreasing in $[z,b]$. $f$ is \emph{single-dipped} if the function $-f$ is single-peaked. For example, (truncated) Normal distributions are single-peaked, and Uniform distributions are weakly single peaked. We denote the cumulative distribution function corresponding to $f$ by $F(X)= Pr_{z\sim f}(z<x)$.

\newpar{Mechanisms}
 A \emph{mechanism} $\g:\cal X^n \rightarrow \cal X$ (also called a voting rule) is a function that maps any profile (set of positions) to a winning position.

Two particular mechanisms we will consider for the interval setting are the \emph{Mean mechanism}, $\mn(S) = \frac{1}{|S|}\sum_{s_i\in S}s_i$, and the \emph{Median mechanism}, $\md(S) = \min\{s_i\in S \text{ s.t. } |\{j:s_j\leq s_i\}| \geq |\{j:s_j> s_i\}|$ (see Fig.~\ref{fig:line}). %

For the binary issues we will focus on a simple \emph{Majority mechanism} that aggregates each issue independently according to the majority of votes. That is, $(\mj(S))^{(j)} = 1$ if $|\{i:s_i^{(j)}=1\}| > |\{i:s_i^{(j)}=0\}|$ and $0$ otherwise, where $s^{(j)}$ is the $j$'th entry of position vector $s$. In all mechanisms we break ties lexicographically towards the lower outcome. 

All of our three mechanisms naturally extend to such infinite populations, as the Median, Mean, and Majority of $f$ (in their respective domains) are well defined. The mechanisms also extend to weighted finite populations.  E.g. for $n$ agents  with positions $S$ and weights $\vec w=\{w_1,\ldots,w_n\}$, the weighted mean is defined as $\mn(S,\vec w) \equiv \frac{1}{\sum_{i
\leq n} w_j}\sum_{i \leq n}w_i s_i$, and similarly for the Median and Majority.  
 
\medskip
In our model, a finite subset $N$ of $n$ \emph{agents} are selected out of the whole population,  and only these agents can vote. We follow \cite{mueller1972representative} in assuming that positions $S_N=\{s_1,\ldots,s_n\}$ are sampled i.i.d. from $f$. We can think of these as voters who happen to be available at the time of voting, or voters for which this voting is important enough to consider participation.

In our basic setup, the unavailable voters abstain, while \emph{all agents vote}. The result is $\g(S_N)$. Yet two problems may prevent us from getting a good outcome. First, $N$ may be too small for $\g(S_{N})$, the decision made by the agents, to be a good estimation of $\g(f)$, the true preference of the population. Second, even selected agents may decide not to vote due to various reasons, and such strategic participation may bias the outcome. We will then have a set of active agents $M \subseteq N$, and the outcome $\g(S_M)$ may be very far from  both $\g(S_N)$ and $\g(f)$, depending on the equilibrium outcome of the induce game (later described in more detail).

%

\newpar{Proxies and weights}
Our main focus in this paper is characterizing the regime in which \emph{voting by proxy} is beneficial. In this setup each inactive voter specifies one of the active agents as a proxy to vote on her behalf. Given a set $M$ of active agents, the decisions of inactive voters are specified by a mapping $J_M:\calX \rightarrow M$, where $J_M(x)\in M$ is the proxy of any  voter located at $x\in \calX$. We label the Proxy setup as $P$, in contrast to the Basic setup denoted as $B$. We highlight that all voters select a proxy, whether they are part of $N$ or not.

Without further constraints, we will assume that the proxy of a voter at $x$ is always its nearest active agent, i.e. the agent whose position (or preferences) are most similar to $x$. Thus for every set $M$, we get a partition (a Voronoi tessellation) of $\calX$ and can compute the \emph{weight} of each active agent $j$ by integrating $f$ over the corresponding cell. Formally, $J_M(x) = \argmin_{j\in M}{\|x-s_j\|}$ and $w_j = \int_{x\in \cal X: J_M(x)= j}f(x)dx$.  
  The outcome of each mechanism $\g$ for agents $N$ is then defined as $\g^B(S_N) = \g(S_N)$ in the Basic scenario,  and $\g^P(S_N) = \g(S_N,\vec w_N)$ in the Proxy scenario, where $\vec w_N$ is computed according to $f$ as above (see Fig.~\ref{fig:line}). The distribution $f$ should be inferred from the context.

\begin{figure}[t]
	\input{line_fig}
	\protect\caption{\label{fig:line}
		The top figure shows the preferences of 4 agents on an interval, as well as the outcomes of the median and mean mechanisms. In the middle figure we see the weight of each agent under proxy selection, assuming $f$ is a uniform distribution on the whole interval, as well as the modified outcomes. The bottom figure shows the outcome under proxy voting if agent~2 becomes inactive, and $M=\{1,3,4\}$. The dotted line marks $\mn(f)=\md(f)=5$. \vspace{-5mm}}
\end{figure}
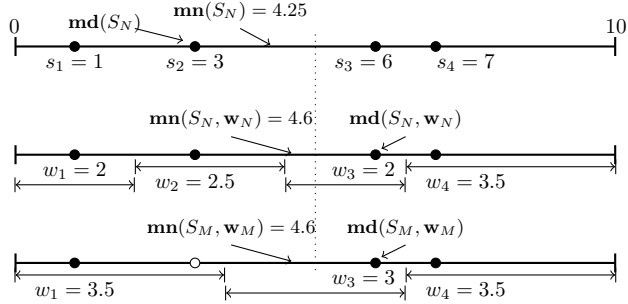 

\newpar{Equilibrium under strategic participation}
In our strategic scenarios the agents $N$ are players in a complete information game, whose (ordinal) utility exactly matches their preferences as voters. I.e., they prefer an outcome that is as close as possible to their own position. Each agent has two actions: active and inactive. In addition, a voter who is otherwise indifferent between the two possible outcomes (i.e. he is not pivotal) will prefer to remain inactive, a behavior known as \emph{lazy-bias}~\cite{elkind2015equilibria}. We refer to these strategic/lazy-bias scenarios by adding ${+}L$ to either $B$ or $P$. Agents may not misreport their position.


When there are no proxies (scenario $B{+}L$) this strategic decision is very simple, since each agent has a single vote which may or may not be pivotal (and when it is pivotal it always helps the agent). On the other hand, if voting by proxy is allowed (scenario $P{+}L$),  any change in the set of active agents changes the proxy selection and thus the weights of all remaining agents. 
Recall that $\vec w_{\vec M}$ denotes the weights we get under proxy selection with active set $M$. 
Then for all $i\notin M$, agent~$i$ prefers to join set $M$ iff $\left\|\vec g(S_{M \cup \{i\}},\vec w_{M \cup \{i\}}) - s_i \right\| < \left\|\vec g(S_M,\vec w_{M}) - s_i \right\|$. 

For example, if agent~2 in Fig.~\ref{fig:line} (bottom) becomes inactive, we get no change in the Median outcome $\md(S_{M},\vec w_M)$, and thus agent~2 prefers to become inactive (it is also possible that an agent strictly loses when becoming active).

A pure Nash equilibrium, or equilibrium for short, is a subset $M\subseteq N$ s.t. no agent in $M$ prefers to be inactive, and no agent in $N\setminus M$ prefers to be active. While it is possible that there are multiple equilibria (or none at all), this will turn out not to be a problem in most cases we consider. \rmr{ok for median and majority, what about mean?}
We thus define $\g^{B{+}L}(S_N) = \g(S_M)$ and $\g^{P{+}L}(S_N) = \g(S_M,\vec w_M)$, where $M\subseteq N$ is the set of active agents in equilibrium.

\medskip
 To recap, an \emph{instance} is defined by a population distribution $f$,  a scenario $Q\in\{\! B,P,B+L,P+L\!\}$, a mechanism $\g\!\in\!\{\!\md,\mn,\mj\!\}$ and a sample size $n$.
We sample a finite profile of $n$  agents i.i.d. from $f$, whose locations are $S_N$. Then, according to the scenario,  either all of $N$ are active, or we get a subset $M$ of active agents. The votes of all active agents are aggregated according to $\g$, with or without being weighted by $\vec w_{\vec M}$, the number of their inactive followers. Finally, the outcome of mechanism $\g^Q(S_N)$ depends on a subset of these parameters, according to the scenario $Q$.
%
\newpar{Evaluation}
We do not consider here the reasons for using one mechanism over another, and simply assume that $\g(f)$ reflects the best possible outcome to the society or to the designer. 
We want to measure how close is $\g^Q(S_N)$ to the \emph{optimal outcome} $\g(f)$. We define the \emph{error} as the  distance between $\g^Q(S_N)$ and $\g(f)$, i.e.,   $\|\g^Q(S_N)-\g(f)\|$. \rmr{It actually does not matter for these 3 domains which $\ell_p$ norm is used}

The \emph{loss} of a mechanism $\g$ is calculated according to its expected error---the expected squared distance from the optimum---over all samples of $m$ available voters.
\labeq{loss}
{\loss^Q(n) =\mathbb{E}_{S_N\sim f^n}\left[ \|\g^Q(S_N)-\g(f)\|^2\right],}
where the mechanism $\g$ and the distribution $f$ can be inferred from the context, and the expectation is over all subsets of $n$ positions sampled i.i.d. from distribution $f$ (sometimes omitted from the subscript). 

 We note that the loss is the sum of two components~\cite{wackerly2007mathematical}: the (squared) bias $\mathbb{E}[\g^Q(S_N) - \g(f)]^2$ and the variance $\var[\g^Q(S_N)]$.
A mechanism $\g$ is \emph{unbiased} for $(Q,f)$ if $\mathbb{E}[\g^Q(S_N)] = \g(f).$
For example in the Basic scenario, mechanisms $\mn$ and $\mj$ are unbiased for $(B,f)$ regardless of $f$, and $\md$ is unbiased for $(B,f)$ if $f$ is symmetric, but not for other (skewed) distributions.

%
%
%

Our primary goal is to characterize the conditions under which proxy voting improves the outcome, i.e.   $\loss^{P[+L]}(n) < \loss^{B[+L]}(n)$. 

\section{Median Voting on an Interval}
\label{sec:median}
The Median mechanism is popular for two primary reasons. First, it finds the point that minimizes the sum of distances to all reported positions, i.e. $\md(S) \in \argmin_{x\in \calX}\sum_{s_i\in S}|\! s_i-x\! |$. Second, in strategic settings where agents might misreport their positions, it is known that the Median mechanism is group strategyproof~\cite{moulin1980strategy}, meaning that no subset of agents can gain by misreporting.
%
%
%

\newsubsec{Random participation}
Suppose all $n$ agents sampled from $f$ are active. 
Let $j^*\in N$ be the proxy closest to $x^*=\md(f)$, and $s^*=s_{j^*}$.

\begin{lemma}\label{lemma:median_optimal} 		$\md^{P}(S_N)= s^*$ for any distribution $f$. 
\end{lemma}
\begin{proof}
	Recall that $\md^{P}(S_N)=\md(S_N,\vec w)$ where $w_j$ is the weight of voters using $j\in N$ as a proxy. 
	All voters $x\geq \md(f)$ are mapped to one of the proxies $j^*,j^*+1,\ldots,n$, thus $\sum_{j=j^*}^n \geq 1/2$ and $\md(S_M,\vec w)\geq s^*$. Similarly, all voters $x\leq \md(f)$ are mapped to one of the proxies $1,2,\ldots,j^*$, thus $\sum_{j=1}^{j^*} \geq 1/2$ and $\md(S_N,\vec w)\leq s^*$. Thus $\md(S_N,\vec w)=s^*$. 
\end{proof}
Thus $\md^{P}(S_N)$ always returns the proxy closest to $x^*=\md(f)$, whereas $\md^{B}(S_N)$ returns some $j\in N$, meaning that the error is \emph{never higher} with proxy voting. I.e., $|\md^P(S)-x^*| \leq |\md^B(S)-x^*|$ for any $S$. In particular, the loss (=expected error) is weakly better.
\begin{corollary}For the Median mechanism, $\loss^{P}(n) \leq \loss^{B}(n)$ for any distribution $f$ and sample size $n$.
\end{corollary} 
\begin{proof}
$$\loss^{P}(n)= E[(\md^{P}(S_N)-x^*)^2] \leq E[(\md^{B}(S_N)-x^*)^2] = \loss^{B}(n).$$
\rmr{ removed since for asymmetric distributions it does not equal the variance 
$$\loss^{P}(n) = \var_{S_N\sim f^n}[\md(S_N,\vec w)] = \var[s_k] $$ \\
$$ = E[(s_k-x^*)^2] \leq E[(\md(S_M)-x^*)^2]  =  \var[\md(S_M)] = \loss^{VA}(\md,n,f).$$
}
\end{proof}

Note that for symmetric distributions, both of $\md^{B}(S_N)=\md(S_N)$ and $\md^{P}(S_N)=\md(S_N,\vec w_N)$ are unbiased from symmetry arguments. Therefore to compute or bound the loss we just need to compute the variance of $\md^{Q}(S_N)$. 
For the unweighted median, this problem was solved by Laplace (see \cite{stigler1973studies} for details): 
Let $x^*=\md(f)$ be the median of symmetric distribution $f$ s.t. $f(x^*)>0$.\footnote{We will assume in this section that $f(x)>0$ in some environment of $x^*$, which is a very weak assumption.} The variance of $\md(S_N)$ is given by (approximately) $\frac{1}{4n f(x^*)^2}=\Theta(1/n)$. Since for any distribution the loss (or MSE) is lower-bounded by the variance, we get that for the Median mechanism, $\loss^B(n)=\Omega(\frac{1}{n})$.

We argue that the loss decreases \emph{quadratically faster} with the number of agents once proxy voting is allowed. 
\begin{conjecture} \label{conj:median}
For the Median mechanism, $\loss^P(n) = O(\frac{1}{n^2})$ for any distribution $f$. 
\end{conjecture}
The rest of this section is dedicated to supporting this conjecture. In particular, we prove it for symmetric distributions, and show empirically that it holds for other distributions as well. Further, for Uniform and single-peaked distributions, we can upper-bound the constant in the expression. 

\begin{theorem} For the Median mechanism, $\loss^{P}(n) = O(\frac{1}{n^2})$, for any symmetric distribution $f$.\label{th:median_bound}
\end{theorem}
\begin{proof}
	W.l.o.g. we can assume $x^*=0$, and that the support of $f$ is the interval $[-1,1]$. 
	What is the expected distance between $s^*$ and $\md(f)=x^*=0$?  We can translate each proxy $x_i$ to $y_i = |x_i-x^*|=|x_i|$. Note that $y_i$ come from some distribution $f'$ on $[0,1]$. By our assumption, $f(x)$ is strictly positive in some $\eps$ environment of $x^*=0$, i.e.  $f(x)>\alpha$ for all $x\in [-\eps,\eps]$, for some $\alpha,\eps>0$. We thus have that $f'(z)>\alpha$ for all $z\leq \eps$. Then for the cumulative distribution $F'(z)$, we have that $F'(z)>z\alpha$ for all $z\leq \eps$, and $F'(z)> \eps\alpha$ for all $z>\eps$. 
	
	Recall that by Lemma~\ref{lemma:median_optimal}, the error is exactly $|s^*-x^*|=|s^*|$.
	
	The random variable $s^*=\min y_i$ is the minimum of $n$ variables sampled i.i.d. from $f'[0,1]$. The distribution of the minimum is well known and in particular for all $z\in[0,1]$,
	$$Pr(s^* > z) 
	= (Pr_{Z\sim f'[0,1]}(Z>z))^n = (1-F'(z))^n.$$
	For $z=\eps$, we get $Pr(s^* > \eps)=(1-F'(\eps))^n \leq (1-\eps\alpha)^n$. 
	
	There is some $n_\eps$ s.t. for all $n>n_\eps$,  $\Pr(s^*>\eps)  < \frac{1}{n^3}$ since the left terms drops exponentially fast. Thus assume $n>n_\eps$. Let $T_n= \floor{2n \cdot \eps}$ 
	
	\begin{align*}
		\loss &^{P}(n) 
		= VAR_f[s^*] 
		= E_f[(s^*)^2]   \leq \sum_{t=1}^{2n} Pr(s^*\in [\frac{t-1}{2n},\frac{t}{2n}])(\frac{t}{2n})^2 \tag{bound by steps}\\
		& = \sum_{t=1}^{T_n} Pr(s^*\in [\frac{t-1}{2n},\frac{t}{2n}])(\frac{t}{2n})^2 +\sum_{t=T_n+1}^{2n} Pr(s^*\in [\frac{t-1}{2n},\frac{t}{2n}])(\frac{t}{2n})^2 \\
		& \leq \sum_{t=1}^{T_n} Pr(s^*\in [\frac{t-1}{2n},\frac{t}{2n}]) (\frac{t}{2n})^2 + \sum_{t=T_n+1}^{2n} Pr(s^*\in [\frac{t-1}{2n},\frac{t}{2n}]) \\
		& \leq \sum_{t=1}^{T_n} Pr(s^*\in [\frac{t-1}{2n},\frac{t}{2n}]) (\frac{t}{2n})^2 + Pr(s^* > \eps) \\
		&  \leq \sum_{t=1}^{T_n} Pr(s^*\in [\frac{t-1}{2n},\frac{t}{2n}]) (\frac{t}{2n})^2 + \frac{1}{n^3} \\
		&= \sum_{t=1}^{T_n} (Pr(s^* > \frac{t-1}{2n})-Pr(s^* > \frac{t}{2n}))(\frac{t}{2n})^2 + \frac{1}{n^3} \\
		&\leq \sum_{t=1}^{T_n} Pr(s^* > \frac{t-1}{2n})(\frac{t}{2n})^2 + \frac{1}{n^3}\\
		&\leq \frac{1}{4n^2} \sum_{t=1}^{T_n}Pr(s^* > \frac{t-1}{2n})t^2 + \frac{1}{n^3}
		\leq \frac{1}{4n^2} \sum_{t=1}^{T_n}(1-F'(\frac{t-1}{2n}))^n t^2 + \frac{1}{n^3}\\
		&\leq \frac{1}{4n^2} \sum_{t=1}^{T_n}(1-\alpha \frac{t-1}{2n})^n t^2 + \frac{1}{n^3}
		\leq \frac{1}{4n^2} \sum_{t=1}^{T_n}e^{-\alpha \frac{t-1}{2}} t^2 + \frac{1}{n^3}\\
		& = \frac{1}{4n^2} \sum_{t=1}^{T_n}e^{-\Theta(t)} t^2 + \frac{1}{n^3}
		 < \frac{1}{4n^2} C + \frac{1}{n^3} = O(\frac{1}{n^2}),
	\end{align*}
	where $C$ is some constant.
\end{proof}

Further, for Uniform $f=U[-1,1]$ we can derive a tighter bound on the sum of the series and show $\loss^{P}(n) < \frac{4}{n^2}$. In fact for any single-peaked $f$ on $[-1,1]$ we have $\loss^{P}(n) < \frac{7}{n^2}$.

	We simulated the effect of proxy voting on the Median mechanism in Figure~\ref{fig:loss_m}. 
	 We can see that the (log of the) loss for each distribution closely resembles $\log(\frac{c}{n^2})  = c' - 2\log(n)$, where the constant $c'$ depends on the distribution. This also holds for the asymmetric distributions, which supports our Conjecture~\ref{conj:median}. In particular, this means that the loss under proxy voting drops much faster than the loss in the Basic scenario, which is roughly $\frac{1}{n}$.
	
	\begin{figure}[t]
		\centering
		\includegraphics[scale=0.42]{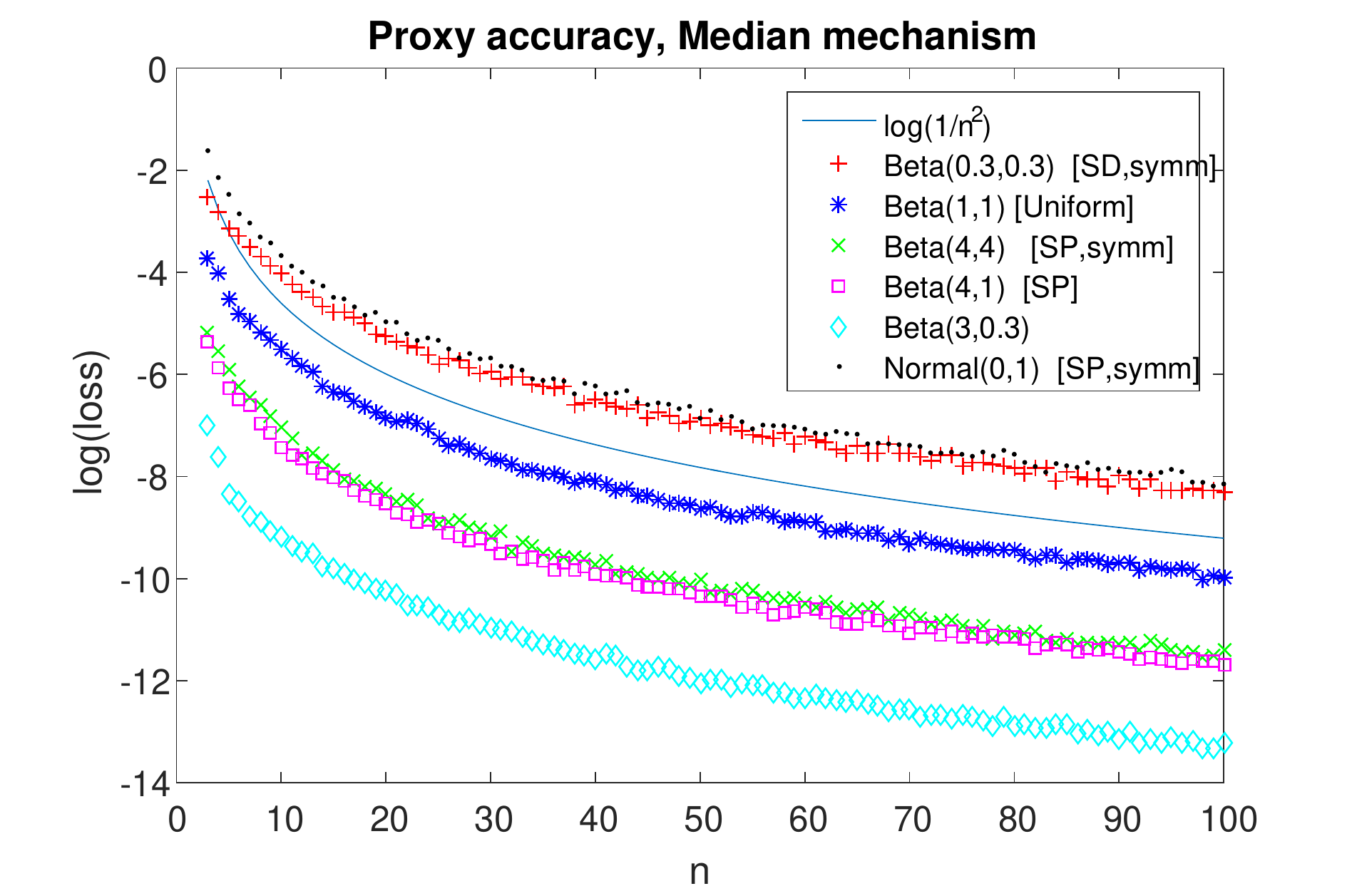} 
		\protect\caption{\label{fig:loss_m}The top figure shows $\loss^Q(n)$ (in log scale) as a function of $n$ for various distributions. SP/SD/symm stands for Single-peaked / Single-dipped/ Symmetric distributions.   Each point is based on 1000 samples of size $n$. 
		\vspace{-4mm} 
		}
		%
	\end{figure}
	%

\newsubsec{Strategic participation}
We show that when participation is strategic the outcome of proxy voting is not affected, whereas the unweighted sample median becomes unboundedly worse. 
Suppose all voters in $N$ are indexed in increasing order by their location, so that $\argmin S_N=1$.
\begin{proposition}
	In the Basic scenario, for any distribution $f$ and any set of agents $N$, there is a unique equilibrium of $\md^{B+L}$ where $M=\{1\}$ (i.e.,  the lowest agent). Further, the game is weakly acyclic, i.e. there is a sequence of best replies from any initial state to this equilibrium.
\end{proposition}
Clearly this means that $\loss^{B+L}(n)<\loss^{B}(n)$, and only gets worse as we increase the sample size $n$.

 The intuition is that due to tie-breaking, either all agents below current median, or all agents above it, are non-pivotal.
\cut{
\begin{proof}
	Note first that if $M=\{1\}$ then the single active agent is pivotal by definition. Any other agent $i>1$ is non-pivotal since $\md(\{s_1,s_i\})=s_1$ by our tie-breaking assumption, thus $M=\{1\}$ is an equilibrium.
	
	Consider any subset of active agents $M\subseteq N$ s.t. $|M|>1$. 
	If $m$ is even, then all agents above the median $\g(M)$ are non-pivotal. If $m$ is odd then all agents below the median are non-pivotal. Thus there is at least one agent in $M$ who prefers to become inactive. This continues until $|M|=1$. 
	
	Finally, if $M=\{i\}$ for some $i>1$, we have the following sequence of best-replies: any agent $j<i$ is pivotal, and in particular $j=1$. Thus agent~$1$ will become active. Now agent~$i$ is no longer pivotal so becomes inactive.
\end{proof}
}

On the other hand, while lazy bias decreases participation in the Proxy scenario, this does not increase the loss. 
\begin{theorem}
	In the Proxy scenario,  for any distribution $f$ and any set of agents $N$, there is a unique equilibrium of $\md^{P+L}$ where $M=\{j^*\}$ (the agent closest to $x^*$). Further, the game is weakly acyclic, i.e. there is a sequence of best replies from any initial state to this equilibrium.
\end{theorem}
In particular, $\loss^{P+L}(n)=\loss^{P}(n)$ for any distribution $f$.
\begin{proof} If $j^*\notin M$ is inactive, then for $M\cup \{j^*\}$ the outcome becomes $s_{j^*}$ rather than $s_k$ (where $k=J_M(x^*)$), which $j^*$ prefers.  If $j^*$ is active, and $j\neq j^*$ quits, then all votes above $s_{j^*}$ are still mapped to $j^*$ or higher (and similarly for votes below $s_{j^*}$). Thus the outcome remains the same which means $j$ is not pivotal. 
\end{proof}


\section{Mean Voting on an Interval}
The Mean mechanism is perhaps the simplest and most common way to aggregate positions. For positions $S_N$ on the interval the outcome is $\mn(S_N)=\frac{1}{n}\sum_{i\in N}s_i$, which is known to minimize the sum of \emph{square distances} to all agents. 
\label{sec:mean}
\newsubsec{Random Participation}
Assume that $f$ is a symmetric distribution, so that $\mn^Q(S_N)$ is unbiased under all scenarios. 
When we apply the Mean mechanism, the loss in the basic scenario is simply the sample variance.


\begin{proposition}
	Let $f$ be a symmetric, weakly single-peaked distribution, and suppose $|N|=2$. Then, for any $S_N$, $\|\mn^P(S_N)-x^*\| \leq \|\mn^B(S_N)-x^*\| $. That is, for any pair of agents the proxy-weighted mean is weakly better than the unweighted mean. 
\end{proposition}
\rmr{We need this proposition in the strategic section to conclude $\loss^{P+L} \leq \loss^{P} < \loss^{B} = \loss^{B+L}$.} 

\begin{proof}
	Suppose w.l.o.g. that the support of $f$ is $[-1,1]$, that $f$ is symmetric around $x^*=0$, that $s_1<s_2$, and that $x=\mn(S_N) = \frac{s_1+s_2}{2} \geq 0$. 
	Then for the basic (unweighted) scenario, 
	$$\|\mn^B(S_N)-x^*\| =|\mn^B(S_N)|=|\mn(S_N)|= |x|=x.$$ 
	Since $f(\cdot)$ in single-peaked, the CDF $F(\cdot)$ is convex in $[-1,0]$ and concave in $[0,1]$, thus for all $z\geq 0$, $F(z)\geq\frac{z+1}{2}$. In particular $F(x) \geq \frac{x+1}{2}$. 
	
	In the proxy (weighted) scenario, agent~1 gets all voters below point $x$, i.e. $w_1=F(x)$, whereas $w_2=1-F(x)$. Thus 
	\begin{small}
		\begin{align*}
	&\mn(S_N,\vec w) = w_1 s_1 + w_2 s_2 = F(\hat x) s_1 + (1-F(\hat x))s_2\\
		&= s_2+F(\hat x)(s_1-s_2) \leq s_2 + \frac{\hat x+1}{2}(s_1-s_2) \tag{as $s_1-s_2<0$}\\
		&= \frac{s_1+s_2}{2} + \hat x\frac{s_1-s_2}{2} = \hat x+ \hat x\frac{s_1-s_2}{2} \\  
		&= \hat x(1+\frac{s_1-s_2}{2})\in  [-\hat x,\hat x].~~\Rightarrow \tag{since $-1 \leq \frac{s_1-s_2}{2} < 0$}\\
		&\|\mn^P(S_N)-x^*\|  = |\mn(S_N,\vec w)| \leq |\hat x| = \|\mn^B(S_N)-x^*\|,
		\end{align*}
	\end{small}
	as required.
\end{proof}

\cut{
\begin{proof}[Proof sketch]
	Suppose w.l.o.g. that the support of $f$ is $[-1,1]$, that $f$ is symmetric around $x^*=0$, that $s_1<s_2$, and that $\hat x=\mn(S_N) = \frac{s_1+s_2}{2} \geq 0$. 
	Then for the basic (unweighted) scenario, 
	$\|\mn^B(S_N)-x^*\| =|\mn^B(S_N)|=|\mn(S_N)|= |\hat x|=\hat x$.
	
	Since $f$ in single-peaked, $F$ is convex in $[-1,0]$ and concave in $[0,1]$, thus for all $z\geq 0$, $F(z)\geq\frac{z+1}{2}$. In particular $F(\hat x) \geq \frac{\hat x+1}{2}$. 
	
	In the proxy (weighted) scenario, agent~1 gets all voters below point $\hat x$, i.e. $w_1=F(\hat x)$, whereas $w_2=1-F(\hat x)$. We can compute the weights and show that $\mn(S_N,\vec w) =(1+\frac{s_1-s_2}{2})\in  [-x,x]$.
	
	This means that $\|\mn^P(S_N)-x^*\|  = |\mn(S_N,\vec w)| \leq |\hat x|$, i.e. weakly better than $\mn^B(S_N)$.
\end{proof}
\rmr{can shorten to a proof sketch}
}

For larger sets of agents this is not true in general. Even for the Uniform distribution there are examples with more agents where proxy voting leads to a less accurate outcome: 
\begin{center}
	\begin{tikzpicture}[scale=0.8,transform shape]
	\def\rr{0.05}
	\draw[thick] (0,0) -- (10,0);
	\draw[fill] (2.5,0.1) circle [radius=\rr];
	\draw[fill] (2.5,0.25) circle [radius=\rr];
	\draw[fill] (10,0.1) circle [radius=\rr];
	\node [above] at (0,0.2) {$0$};
	\node [above] at (10,0.2) {$1$};
	\node [above] at (4.5,0.0) {$\mn(S_N)=\frac12$};
	\node at (5,0.05) {$*$};
	\node [below] at (6.5,0.05) {$\mn(S_N,\vec w_N)=\frac{17}{32}$};
	\node at ( 5.3125,-0.08) {$*$}; 
	\end{tikzpicture}
\end{center}
Consider 3 agents on $\calX=[0,1]$, located at $S_N=\{\frac{1}{4},\frac{1}{4},1\} $.  For a Uniform distribution $f$, the optimal outcome is $ x^* = \mn(f)=\frac{1}{2}$. In the Basic scenario, $\mn^B(S_N)= \frac{1}{2}$  while with proxies,\\
$\mn^P(S_N) = \mn^P(S_N,\vec w_N=\{ \frac58,0,\frac38 \}) = \frac{1}{4}\frac58+1\frac38=\frac{17}{32} $.


The question is under which  
distributions $f$ the loss is improved \emph{on average} by weighing the samples. We show analytically that this holds for uniform distributions and provide similar simulation results for other distributions.

\newpar{Uniform distribution}

Consider the uniform distribution over
the interval $[-1,1]$ (w.l.o.g., as we can always rescale). 
In the Basic scenario, we know from \cite{arnold1992first} that  $\loss^{B}(n)=\var[\mn(S_{N})]=\frac{1}{3n}$. 
The next proposition indicates that the loss under proxy voting decreases quadratically faster than without proxies (as with the median mechanism).

\begin{proposition} \label{th:w_mean_var}
	For the Mean mechanism, when $f=U[-1,1]$,\\
	$\loss^P(S_N) =\frac{8}{n^{2}}(1+O(\frac{1}{n}))$.
\end{proposition} 

\begin{proof}[Proof]
	We first note that the weighted mean $\mn\left(S_{N},\mathbf{w}_N\right)$,
	is an unbiased estimator of the distribution mean from symmetry argument, and therefore $\loss^P(S_N) = \var[\mn(S_{N},\vec{w}_N)]$.
	We now turn to evaluate this term. $\var[\mn(S_{N})]=\mathbb{E}\left[\left(\mn(S_{N})\right)^{2}\right]$.

	\[	\mn\left(S_{N},\mathbf{w}_N\right)=\frac{1}{n}\sum_{j=1}^{n}w_{j}s_{j} 	\]
	
	Here $w_{j}$ is the number of voters that elect representative
	$j$ as their proxy. Since the number of vote $n$ is large, $w_{\alpha}$
	is the corresponding share of the probability distribution, 
	In the Uniform distribution $U[-1,1]$ we can compute the weights:
\begin{align*}
	w_{j}&=F\left(\frac{s_{j+1}+s_{j}}{2}\right)-F\left(\frac{s_{j-1}+s_{j}}{2}\right)\\
	&=\frac12 (\frac{s_{j+1}+s_{j}}{2}-(-1)) + \frac12(\frac{s_{j-1}+s_{j}}{2}-(-1))=
	\frac{1}{4}\left(s_{j+1}-s_{j-1}\right)
	\end{align*}

	where we set $s_{0}=-2-s_{1}$, $s_{n+1}=2-s_{n}$ for convenience.
	Therefore $\mathbf{mn}\left(S_{N},\mathbf{w}_N\right)$ can be written as
	
	\labeq{w_avg}
	{\sum_{i=1}^{n}w_{j}s_{j}  =  \frac{1}{4}  \sum s_{j}\left(s_{j+1}-s_{j-1}\right)
		 =  \frac{s_{n}+s_{1}}{2}+\frac{s_{1}^{2}-s_{n}^{2}}{4}
	}
	by telescopic cancellation. Here $s_{n}$ and $s_{1}$ are the two
	extremes representatives. Now, since the joint distribution of $(s_{1},s_{n})$
	is explicitly known~\cite{arnold1992first},
	\[
	\Pr\left(s_{1}=x,s_{n}=y\right)=n\cdot\left(n-1\right)\cdot\left(\frac{1}{2}\right)^{2}\cdot\left(\frac{y-x}{2}\right)^{n-2}
	\]

	it is possible to evaluate it precisely by integration. We get \rmr{I divided Eli's result by 4. hope it's ok}
	\begin{align*}
	 & \mathbb{E}\left[\mathbf{mn}\left(S_{N},\mathbf{w}\right)^{2}\right] \\ 
	&=\frac{2((n-5)n+14)n\cdot\left(n-1\right)}{\prod_{t=1}^6(n+t)}\\
	 &= \frac{2}{n^{2}}\left(1+O(n^{-1})\right). \\
	 &= \frac{2}{n^{2}}\left(1+O(\frac{1}{n})\right)  . 
	\end{align*}
	
	We should note that the estimator $\frac{s_1+s_n}{2}$ is known to minimize the MSE for the uniform distribution. It is interesting that the estimator obtained by proxy voting is so similar.
	\rmr{The bound in the proposition is obtained by rescaling}
\end{proof}

\cut{
Recall that  $\frac{s_1+s_n}{2}$ is the maximum likelihood estimator of $\mathbb{E}[f]=\mn(f)$ for the uniform distribution, and
$\var\left[s_{1}+s_{n}\right] =\frac{2}{n^{2}}\left(1+O(\frac{1}{n})\right)$, which means that $\mathcal{L^P}/\mathcal{L^B}\rightarrow 4$.}

%
%
%
%
%
%

While proxy voting may have adverse effect on the mean in specific samples, our proof shows that \emph{on average}, proxy voting leads to a substantial gain under the Uniform distribution. Other common distributions displayed the same effect.  Fig.~\ref{fig:loss_mn} shows proxy voting leads to a substantial improvement over the unweighted mean of active voters for various distributions.

\begin{figure}[t]
	\centering
	\includegraphics[scale=0.42]{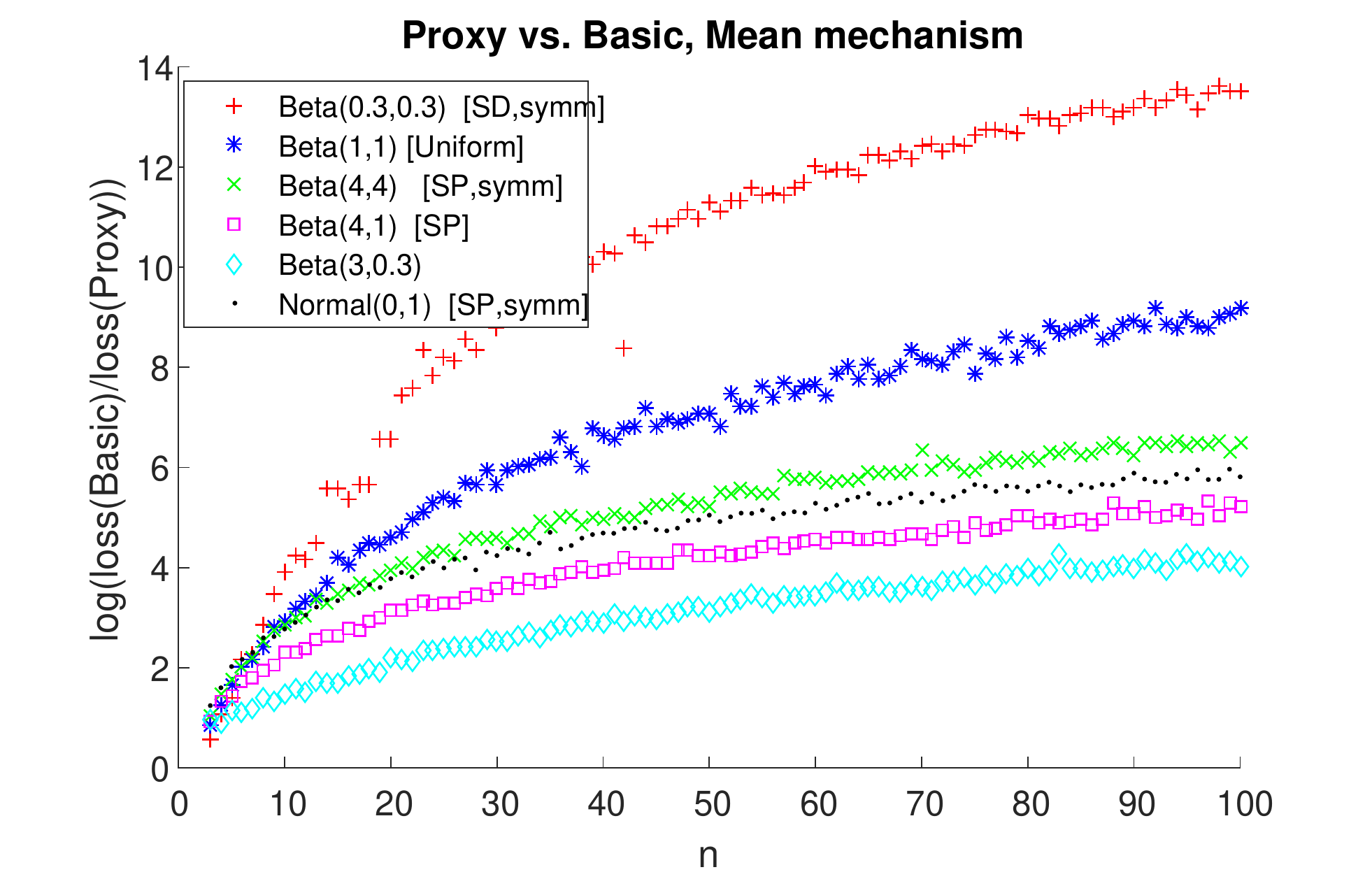} 
	\protect\caption{\label{fig:loss_mn}The ratio of $\loss^B(n)$ and $\loss^P(n)$ (in log scale) as a function of $n$.
	\vspace{-4mm}}
\end{figure}

\newsubsec{Strategic participation}
In the basic (non-proxy) scenario, it is easy to see that every voter is always pivotal with any active set unless $s_i=\mn(M)$. Thus in every equilibrium $M\subseteq N$, $\mn(S_M)=\mn(S_N)$, and for any distribution $f$, and $\loss^{B+L}(n)=\loss^B(n)$.

In the proxy setting things get more involved. The following lemma analyzes the best response of agents in cases where the voter's population is monotonic is some region.

\begin{lemma}
	\label{lem:mean-strategic}
	%
	(A) It is a dominant strategy for both  $\argmin_i\left\{ S_{N}\right\} $ and $\argmax_i\left\{ S_{N}\right\}$ to be active; 
	
	(B) Consider three agents, $s_{1}<s_{2}<s_{3}$ s.t. $\{1,3\}\subseteq M$. Suppose $f$ is strictly decreasing in $[s_1,s_3]$. Agent~2 prefers to be active if $\mn\left(S_{M \cup \{s_2\}}\right)<s_{2}$, and prefers to be inactive if $\mn\left(S_{M \setminus \{s_2\}}\right)\geq s_{2}$. The reverse condition applies for increasing $f$. If $f$ is constant, agent~2 always prefer to be inactive. 
\end{lemma}

\begin{proof}
(A) is obvious. 
	For (B), consider a set of active agents $M^{-}$ such that $\{s_{1},s_{3}\}\subseteq M^{-}$
	and $s_{2}\notin M^{-}$. Define $M^{+}=M^{-}\cup\{s_{2}\}$. 
	The population decision boundary
	points are the intermediate points between the different agents $\alpha=\left(s_{1}+s_{2}\right)/2$, $\gamma=\left(s_{2}+s_{3}\right)/2$ and $\beta=\left(s_{1}+s_{3}\right)/2$. The result of the decision mechanism is denoted as $g_{M^{-}}$,
	or correspondingly, $g_{M^{+}}$. We note that $\mathbf{mn}\left(S_{M^{+}}\right)-\mathbf{mn}\left(S_{M^{-}}\right)$ equals\\ $\left(s_{2}-s_{1}\right)\int_{\alpha}^{\beta} f(x)d(x)+\left(s_{2}-s_{3}\right)\int_{\beta}^{\gamma} f(x)d(x)
	$.

	By the intermediate value theorem, there exists a point $x_{1}\in[\alpha,\beta]$
	such that 
		$\int_{\alpha}^{\beta}f(x)d(x)  =  \left(\beta-\alpha \right)f(x_{1})
		 =  \frac{(s_{3}-s_{2})f(x_{1})}{2}$.

	Similarly, there is $x_{2}\in[\beta,\gamma]$ such that \\
	$
	\int_{\beta}^{\gamma}f(x)d(x)=\frac{(s_{2}-s_{1})f(x_{2})}{2}
	$.
	Therefore, 
	\begin{small}
	\begin{equation}
	\mathbf{mn}(\! S_{M^{+}}\!)-\mathbf{mn}(\! S_{M^{-}}\! )=\frac{\left(s_{2}\!-\!s_{1}\right)\left(s_{3}\!-\!s_{2}\right)\left(f(x_{1})\!-\!f(x_{2})\!\right)}{2}\label{eq:condition}
	\end{equation}
	\end{small}
	
	This expression is positive if $f\left(x_{1}\right)>f\left(x_{2}\right)$.
	If $f$ is monotonic decreasing in $[s_{1},s_{3}]$ this holds, while
	if $f$ is monotonic increasing we have $\mathbf{mn}\left(S_{M^{+}}\right)<\mathbf{mn}\left(S_{M^{-}}\right).$
	 If $\mathbf{mn}\left(S_{M^{-}}\right)<s_{2}$ and $f$ is increasing, it is not beneficial for $s_{2}$ to become active. Likewise, if $\mathbf{mn}\left(S_{M^{-}}\right)>s_{2}$
	and $f$ is monotonic decreasing $s_{2}$ will not be active.
	Finally, if $f$ is constant, then $f(x_1)=f(x_2)$ and agent $s_2$ does not affect the result and will be inactive.
%
%
%
%
%
%
%
%
%
\end{proof}

Before considering general probability distributions, we apply the
previous lemma for the particular case of the uniform distribution.
We show that even when the voters are strategic, the result equilibrium
is the optimal configuration.
\begin{proposition}\label{prop:uni_mean_eq}
In the Proxy scenario, for the Uniform distribution and any set of agents $N$, there is a unique equilibrium of $\mn^{P+L}$ where $M=\{\argmin S_N,\argmax S_N\}$ (i.e., the two extreme agents). Further, the game is weakly acyclic, i.e. there is a sequence of best replies from any initial state to this equilibrium.
 \end{proposition}
\begin{proof}
	Lemma~\ref{lem:mean-strategic}(A) says it is beneficial
	that the two most extreme agents to be active. Due to Part~(B), all other agents will quit. 
\end{proof}
Our last result for uniform distributions shows that strategic behavior, despite lowering the number of active agents, leads to a more accurate outcome than in the non-strategic case. In fact, it can be shown that \emph{no other estimator} outperforms $\mn^{P+L}$ for the Uniform distribution.

\begin{corollary}
For the Mean mechanism, for any sample $S_N$, under the unique equilibrium of $\mn^{P+L}$ for Uniform $f$, 
$\mn^{P+L}(S_N) = \mn^P(S_N)$.  In particular, $\loss^{P+L}(n)= \loss^{P}(n)$.
\end{corollary}
  
\cut{
\begin{proof}
w.l.o.g. $f=U[-1,1]$. 
For any sample $S_N$, let $S_M=\{s_1,s_n\}$ contain the two extreme samples. Let $\vec w_M= (w^*_1,w^*_n)$ denote the weights of these samples under proxy voting, when there are no other agents. 
We have that 
\begin{align*}
	\mn^{P+L}(S_N)&= \mn(S_M,\vec w_M)= \frac12 (s_1 w^*_1 + s_n w^*_n) \\
	&= \frac{1}{2}\ensuremath{\left(s_{1}\ensuremath{\left(\frac{s_{1}+s_{n}}{2}+1\right)}+s_{n}\ensuremath{\left(1-\frac{s_{1}+s_{n}}{2}\right)}\right)}
	=  \frac{s_{1}+s_{n}}{2}+\ensuremath{\frac{s_{1}^{2}-s_{n}^{2}}{4}}\\
	&= \mn(S_N,\vec w_N) = \mn^{P}(S_N), \tag{by Eq.~\eqref{eq:w_avg}}
\end{align*}
as required.
\end{proof}
}

\begin{figure}

	\begin{centering}
		\includegraphics[width=1\columnwidth]{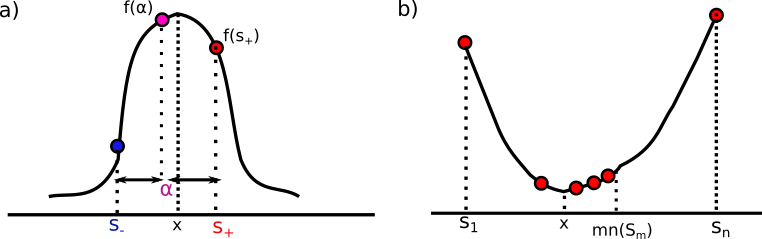}\protect\caption{\label{fig:single-peak-and-dip}a) An equitable partition of a a single 
			peaked distribution. b) A feasible equilibrium in a single dip scenario.
			Note that in both cases the distribution is not necessarily symmetric.}
		
		\par\end{centering}
	
\end{figure}

\medskip
We now turn to analyze more general distributions, and we first focus on the single peak case. Denote the
peak location as $x$, the smallest agent in $[x,\infty]$ as $s_{+}$
and the largest agent in $[-\infty,x]$ as $s_{-}$. Set $\alpha=\left(s_{+}+s_{-}\right)/2$ as the intermediate point between $s_{+}$ and $s_{-}$. Assume, w.l.o.g, that $f\left(s_{-}\right)\geq f\left(s_{+}\right)$. We call a given set of agents $M$ an \emph{equitable partition} if $\mn^P(S_M)\in[s_-,s_+]$  and $f(A)\geq f(s_{+})$ (Fig. \ref{fig:single-peak-and-dip}(a)).

\begin{proposition}
	Consider a single peaked distribution $f$ and a profile $S_{N}$. If $S_{N}$ is an equitable partition, then there is an equilibrium of $\mn^{P+L}$ where all agents
	are active $M=N$. In particular, the error is the same as in $\mn^{P}$.
	\end{proposition}
\begin{proof}
	Consider the set $M=N$. Following Eq.~\ref{eq:condition}, 
	$s_{+}$ will not quit from the active set if 
	$\left(s_{2}-s_{1}\right)\left(s_{3}-s_{2}\right)\left(f(x_{1})-f(x_{2})\right)>0$.

	We shall now show that $f(x_{1})>f(s_{+})$. Assume $A<x$. As $f(A)\geq f(s_{+})$,
	for every $y\in[A,x]$ we have $f(y)>f(A)\geq f(s_{+})$ as $f$ is
	increasing in $[A,x]$. Likewise, $f(y)>f(s_{+})$ as $f$ is decreasing
	in $[x,s_{+}]$ therefore $f(x_{1})$, the mean value of $f$ in $[A,s_{+}]$
	satisfies $f(x_{1})>f(s_{+})$. Now, $f(s_{+})>f(x_{2})$ as $f$
	is monotonic decreasing in $[x,\infty]$. Therefore, the former expression
	is positive, and $s_{+}$ will stay in the equilibrium set. 
	
	Now, Lemma \ref{lem:mean-strategic}(A) shows that the most extreme
	agents $s_{1},s_{n}$ will be active, while Lemma\ref{lem:mean-strategic}
	(B) shows every agent between $[s_{+},s_{n}]$ and $[s_{1},s_{-}]$
	are also active. Namely, all the agents are active.
\end{proof}

This shows that proxy voting may achieve maximal participation in a single peak setup. Next, we address the single dip setting. 
\begin{proposition}
	Consider a single dipped distribution where the dip location is $x$. Consider
	any equilibrium $M\subseteq N$, and assume w.l.o.g that
	$\mathbf{mn}\left(S_{M}\right)\leq x$. Then, $M$ contains at
	most two agents in $[\min\left(S_{N}\right),\mathbf{mn}\left(S_{M}\right)]$
	and at most two agents in $[\mathbf{mn}\left(S_{M}\right),\max\left(S_{N}\right)]$.\end{proposition}
	\cut{
\begin{proof}
	Lemma \ref{lem:mean-strategic}(A) shows that the two most extreme
	agents $s_{1},s_{n}$ are always active. Denote the dip
	location as $x$. Consider some active agents set $M$, and assume $\mathbf{mn}\left(S_{M}\right)\leq x$.
	Lemma \ref{lem:mean-strategic}(B) shows that there can not be more
	than two agent in $[x,s_{n}]$ and $[s_{1},\mathbf{mn}\left(S_{M}\right)]$.
	Consider an equilibrium set that contain active agents in $\mathcal{A=}[\mathbf{mn}\left(S_{M}\right),x]$.
	Denote the maximal active agent in $\mathcal{A}$ as $y$. Then Lemma  \ref{lem:mean-strategic}(B) indicates that all agents in $[\mathbf{mn}\left(S_{M}\right),y]$
	are active, while there is only one active agent in $[s_{1},\mathbf{mn}\left(S_{M}\right)]$,
	which is $s_{1}.$
	
	If there are no active agent in $\mathcal{A}$, then Lemma \ref{lem:mean-strategic}(B) show that are at most two agents in $[s_{1},\mathbf{mn}\left(S_{M}\right)]$
	and in $[s,s_{n}]$.
\end{proof}
}
We see the possible emergence of four active agents, or parties, at
the center-right, center-left, extreme right and extreme left. If
the  distribution is heavily skewed, we expect some parties
to emerge between the dip location and the decision rule, balancing
the result.

%
%
%
%
\section{Binary Issues}
\label{sec:binary}
In this section $\cal X = \{0,1\}^k$ and $\mj(S)$ outputs a binary vector according to the majority on each issue.
In the most general case, $f$ can be an arbitrary distribution over $\{0,1\}^k$. However, we assume that issues are conditionally independent in the following way: first a number $P$ is drawn from a distribution $h$ over $[0,1]$, and then the position on each issue is `1' w.p. $P$. That is, the position of a voter on all issues is $(s^{(1)},\ldots,s^{(k)})$, where $s^{(j)}$ are random variables sampled i.i.d from a Bernoulli distribution $Ber(P)$, and $P$ is a random variable sampled from $h$. Since $h$ induces $f$ we sometimes use them interchangeably. 

\newpar{Evaluation} 
W.l.o.g. denote the majority opinion on each issue as 0, meaning that $x^* = \mj(f) = (0,0,\ldots,0)$.  The expected rate of `1' opinions is $\mu\equiv \mathbb{E}_{P\sim h}[P]< 0.5$.  
One interpretation of this model is that $\vec{0}$ is the ground truth, and $P_i$ is the probability that agent $i$ is wrong at any issue. Under this interpretation $1{-}\mu$ is the signal strength that the population has on the truth.  In the lack of ground truth,  the majority opinion is considered optimal.  Here $P_i$ is the probability that agent~$i$ disagrees with the majority at each issue.
The error of a given outcome $z\in \{0,1\}^k$ is then  $\|z-x^*\|=\sum_{j=1}^k z^{(j)}$ (coincides with the Hamming distance between $z$ and $x^*$). The loss is the expected error over samples as before. 

We argue that when society has \emph{limited information} (small sample size $n$ and high mistakes probability $\mu$), then scenario $P$ does better than scenario $B$, i.e. $\loss^{P}(n) < \loss^{B}(n)$.

\newsubsec{Random participation}
Suppose that each agent is wrong w.p. \emph{exactly} $\mu<0.5$, i.e. $t_i\sim Ber(\mu)$ is the opinion of agent~$i$ on a particular issue. Then the probability that the majority is wrong on this issue is $\Pr(\sum_i t_i  > \frac{n}{2})$, as stated by the \emph{Condorcet Jury Theorem}. In our case, $t_i \sim Ber(P_i)$, where $P_i$ differs among agents, and this case of independent heterogeneous variables was covered in \cite{grofman1983thirteen}, which showed: 
\labeq{grofman}
{
\Pr(\sum_i t_i  > \frac{n}{2}) = \Pr(Z_{\mu,n} > \frac{n}{2}),
}
where $Z_{\mu,n}\sim Bin(\mu,n)$ and $\mu=\frac{1}{n}\sum_i P_i$. Since the loss is additive along issues, $\loss^B(n) =  k\Pr(Z_{\mu,n} > \frac{n}{2})$.
%

We now turn to analyze the Proxy scenario. 
Assume w.l.o.g. that $P_1,P_2,\ldots,P_n$ are sorted in increasing order.
As $k$ increases, $P_i$ provides a good prediction of how many 1's and 0's will be in $s_i$. This  enables us to predict how inactive agents will select their proxies: an agent with parameter $P_i<0.5$ will almost always select agent~$1$ and an agent with $P_i>0.5$ will select agent~$n$ w.h.p. 
\begin{lemma} \label{lemma:dictator} For every position $z<0.5$, $\Pr(\exists j\in N \text{ s.t. } \|s_j-z\| < \|s_{1}-z\| ) < n\cdot e^{-bk}$.  for some constant $b$. The same holds for $z>0.5$ and $s_n$. 	
\end{lemma}
\begin{proof}
	Note that $P_{1} < P_j$ for all $j>1$. In addition, we denote a topic disagreement indicator $I_{i,j}^{(l)} =\ind{s_i^{(l)} \ne s_j^{(l)}}$.  For each agent $i$ with $P_i <0.5, \forall j > 1$, 
	\begin{align*}
	&\Pr(\|s_i-s_j\| < \|s_i-s_{[1]}\|) = \Pr(\sum_{l=1}^{k}I_{i,j}^{(l)} < \sum_{l=1}^{k}I_{i,1}^{(l)} ) \\
	&= \Pr(\sum_{l=1}^{k}I_{i,j}^{(l)} - \sum_{l=1}^{k}I_{i,1}^{(l)} < 0 ) \\	
	\end{align*}
	define
	\begin{align*}
	&q_1 = P_i(1-P_{1})+(1-P_i)P_{1}\\
	&q_2 = P_i(1-P_j)+(1-P_i)P_j  
	\end{align*}
	
	Since $P_i<0.5, P_{1} > P_j \Rightarrow q_1 > q_2$
	\begin{align*}
	&X_1 = \sum_{l=1}^{k}I_{i,1}^{(l)} \sim Binomial(k,q_1) 	\\
	&X_2 = \sum_{l=1}^{k}I_{i,j}^{(l)}   \sim Binomial(k,q_2)  \\
	&\Pr(\sum_{l=1}^{k}I_{i,1}^{(l)}< \sum_{l=1}^{k}I_{i,j}^{(l)} ) = \Pr(X_1 - X_2 <0 )\\
	\end{align*}
	Since $k \to \infty $ and $q_1,q_2$ are constants, a normal approximation to binomial distribution will be sufficiently accurate for our purpose. 
	\begin{align*}
	&X_1 \approx Z_1 \sim N(kq_1,kq_1(1-q_1)) \\
	&X_2 \approx Z_2 \sim N(kq_2,kq_2(1-q_2)) \\	 
	& (Z_1 - Z_2) \sim N(k(q_1-q_2),k(q_1(1-q_1)+q_2(1-q_2))) \\
	& \Pr(X_1 -X_2 < 0 ) \approx \Pr(Z_1 -Z_2 < 0 ) \\
	&= \Phi(\dfrac{0 - k(q_1-q_2)}{\sqrt{k(q_1(1-q_1)+q_2(1-q_2))}}) \\
	&=\Phi(\dfrac{\sqrt{k}(q_1-q_2)}{q_1(1-q_1)+q_2(1-q_2)})= \Phi(-a\cdot \sqrt k)  \\
	\end{align*}
	
	for some positive constant $a$. Note that for $x<-1$, $\Phi(x) < O(e^{-\frac{x^2}{2}})$, thus $\Pr(X_1>X_2) < e^{-bk}$ for some constant $b>0$. By the union bound, $\Pr\left(\exists j\in M \text{ s.t. } \|s_j-z\| < \|s_{1}-z\| \right) \leq (m-1)Pr(X_1>X_2) < me^{-bk}$.
\end{proof}

This means that when there are many issues, all voters with $z<0.5$ will cast their votes to agent~1, thus $w_1=\Pr_{z\sim h}(z<0.5), w_n=\Pr_{z\sim h}(z>0.5)$. Hence one of the agents $\{1,n\}$ \textbf{is effectively a dictator}, depending on whether the median of $h$ is below or above $0.5$. From now on we will assume that agent~1 is the dictator, as this occurs with high probability as $k\rightarrow \infty$ under most distributions with $\mu<0.5$.  Thus (for sufficiently large $k$), 
\labeq{dictator} {
    \lVert \mj^P(S_N)-x^* \rVert = k\min_{i \in N}{(P_i)}=kP_1 .}

To recap, under scenario $B$ the majority mechanism is equivalent to unweighted majority of a size $n$ committee, while under scenario $P$, the mechanism is equivalent to a dictatorship of the best expert (i.e., the most conformist agent).

Given a particular distribution $h$, we can calculate $\loss^P(n)$ analytically or numerically. E.g. when $h=U[0,a]$ (note $a=2\mu)$,
$$\loss^{P}(n)= k\mathbb{E}_{P^n\sim U(0,a)^n}[\min_{i \in N}{P_i}] = \frac{ka}{n+1}= \frac{2\mu k}{n+1}.$$

\begin{figure}[t]
	\centering
	\includegraphics[scale=0.32 ]{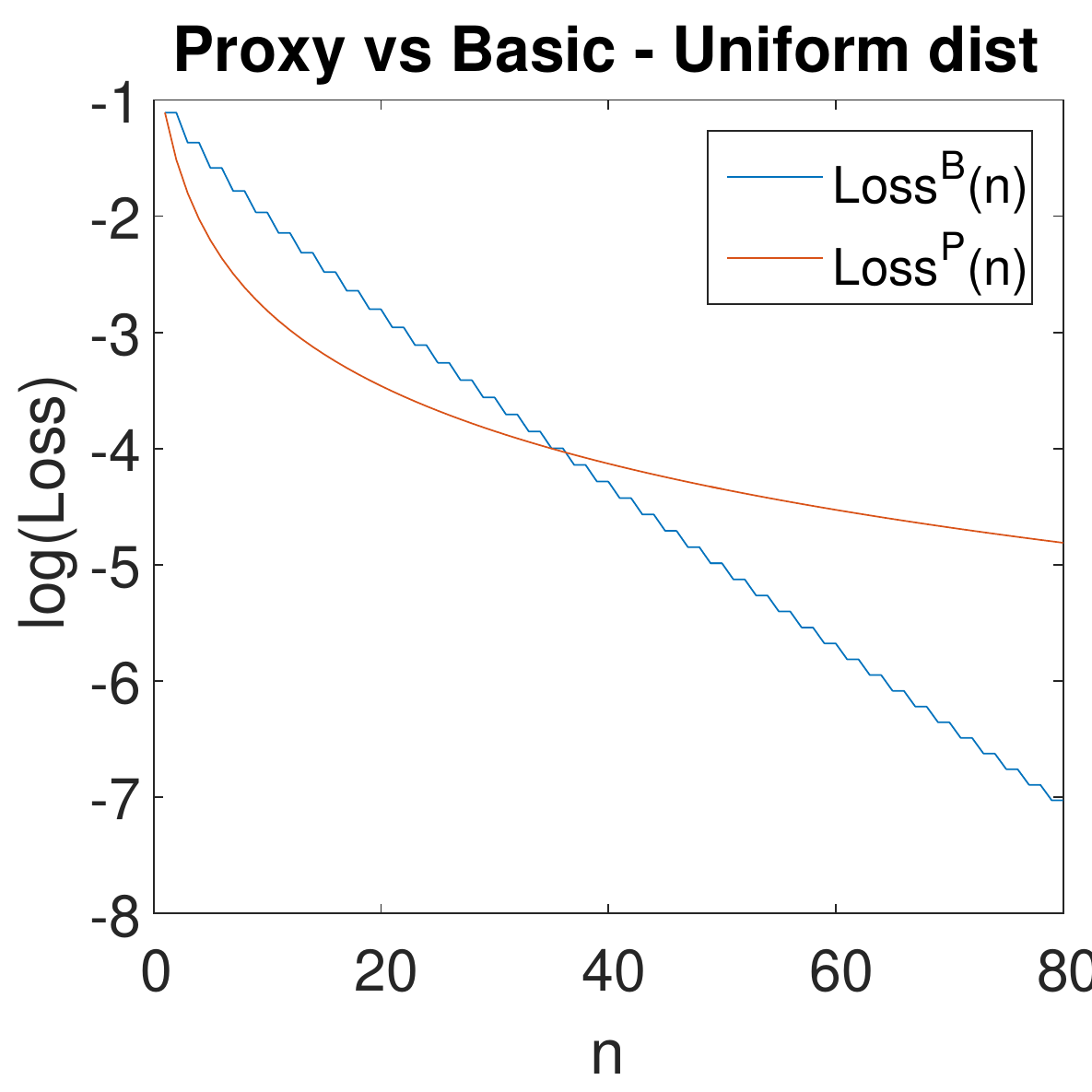}
	\includegraphics[scale=0.32 ]{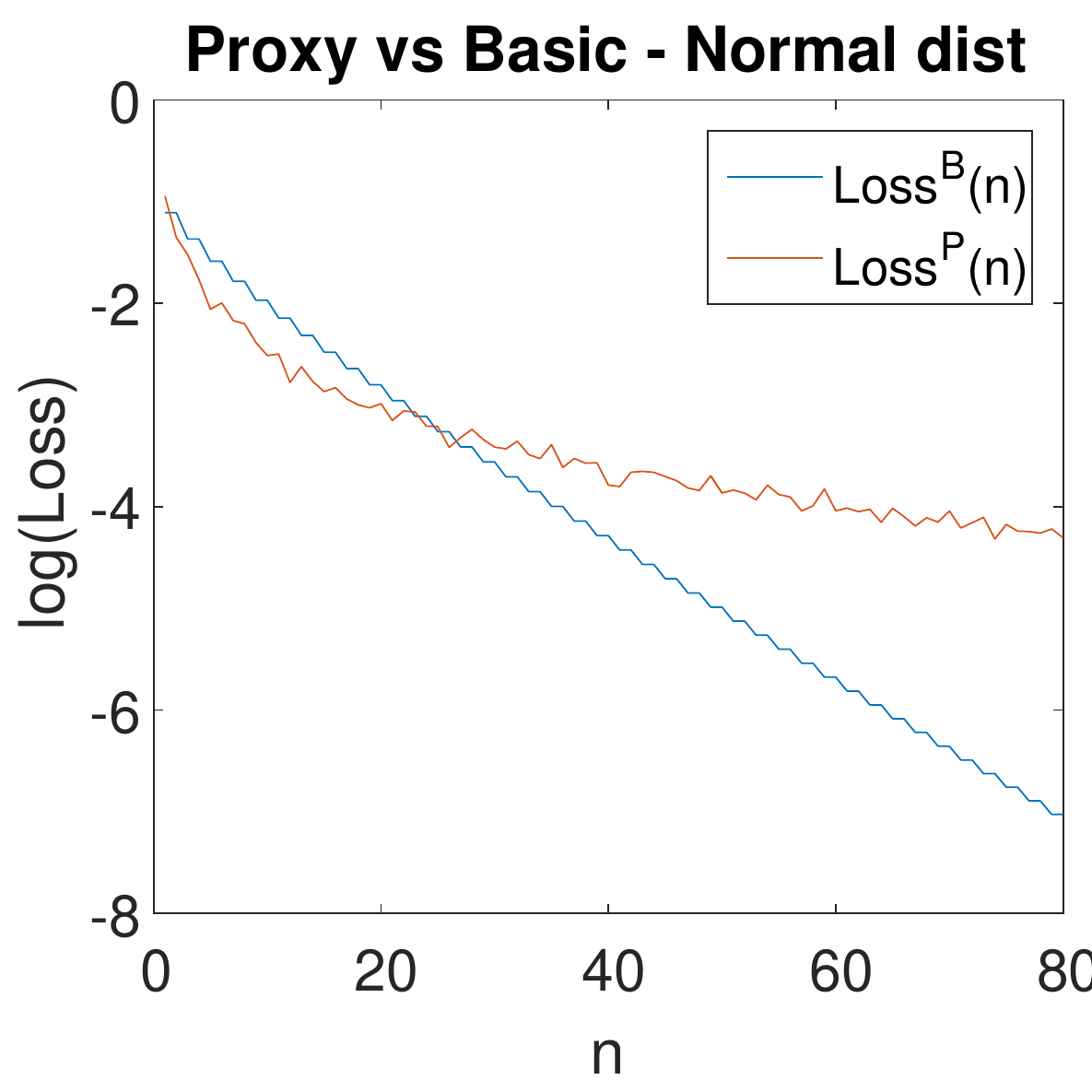}
	\protect\caption{The loss $\loss^Q(n)$ (in log scale), for  distributions $h=U[0,2\mu=0.66]$ (left);  $h=N(\mu=0.33,\sigma=0.3)$ (right).
		\label{fig:binary_theory}\vspace{-4mm}}
\end{figure}

We can infer from Eqs.~\eqref{eq:grofman} and \eqref{eq:dictator} that the proxy voting is beneficial in cases where the best expert out-performs the majority decision on average. Specifically, when the sample is small and/or the signal of agents is weak ($\mu$ is close to $0.5$). See Fig.~\ref{fig:binary_theory}.
%

%

\newsubsec{Strategic participation} 
In general there may be multiple equilibria that are difficult to characterize, and whose outcomes $\mj(S_M)$ may be very different from $\vec x^*$. 
However we can show that for a sufficiently high $k$, there is (w.h.p) only one equilibrium outcome in each of the mechanisms $\mj^{B+L},\mj^{P+L}$. 

Intuitively, the reason is as follows. For every agent $i\in N$ there is w.h.p an issue for which she is pivotal, and thus the only equilibrium in scenario $B{+}L$ will be $M=N$ (w.h.p). 
In scenario $P{+}L$, the entire weight is distributed between the \emph{active agents} with the lowest and highest $P_i$. This means that the best agent is always pivotal and thus active. Regardless of which other agents become active, we get that $\mj^{P+L}(S_N) = \mj(S_M,\vec w_M) = s_1 = \mj^{P}(S_N)$. 
The probability that any other equilibrium exists and affects the loss goes to zero.

\paragraph{Basic setting}
For any $M\subseteq N$, denote by $Y_M$ the event that set $M$ is an equilibrium in the game $\mj^{B+L}(S_N)$. We bound the probability that $N$ is not the unique equilibrium. 
\begin{lemma}
	$\Pr(\neg Y_N \vee (\exists M\subsetneq N, Y_M)) < e^{2n-\frac{k}{2^{n}}}$. Note that for $k \gg n \cdot 2^{n+1}$ the bound tends to $0$.  
\end{lemma}
\begin{proof}
For a binary vector $\vec q\in \{0,1\}^n$, we denote by $Z_{\vec q}$ the event that for some issue $j\leq k$, $q_i=s_i^{(j)}$ for all $i\in N$. We also denote $Z^* = \bigcup_{\vec q\in \{0,1\}^n}Z_{\vec q}$.
   
We first argue that $Z^*$ entails both $Y_N$ and $\neg Y_M$ for any $M\subsetneq N$. 
Consider first the set $N$, and voter $i\in N$. If $n$ is odd consider some vector $\vec q$ where $q_i=1$ and all other voters split evenly between $0$ and $1$. Since $Z^*$ holds, there is an issue $j$ s.t.  $q_{i'}=s_{i'}^{(j)}$ for all $i'\in N$. We get that $\mj(N)^{(j)}=1$ but $\mj(N \setminus \{i\})^{(j)}=0$, i.e. $i$ is pivotal and will thus not quit. If $n$ is even we proceed in a similar way except $q_i=0$ and all of $N$ split evenly between $0$ and $1$.

For any smaller set $M$, consider some $i\in N\setminus M$, where $|M|=m$. If $m$ is \emph{even} we consider a vector $\vec q$ where $q_i=1$ and and all voters \emph{in $M$} split evenly between $0$ and $1$. We get that there is an issue $j$ where  $\mj(M)^{(j)}=0$ but $\mj(N \cup \{i\})^{(j)}=1$, i.e. $i$ is pivotal and will join ($M$ is not stable). 
If $m$ is odd we proceed in a similar way except $q_i=0$ and all of $M \cup \{i\}$ split evenly between $0$ and $1$.

\medskip
It is left to bound $Pr(\neg Z^*)$.
Indeed, for any $\vec q$ and $j\leq k$, the probability that $\vec q= s^{(j)}$ is exactly $2^{-n}$, and thus
\begin{align*}
Pr(\neg Z^*)& \leq \sum_{\vec q}Pr(\neg Z_{\vec q}) = \sum_{\vec q \in \{0,1\}^n}\prod_{j\leq k}Pr(s^{(j)} \neq Z_{\vec q})\\
&=  \sum_{\vec q \in \{0,1\}^n}\prod_{j\leq k}(1-2^{-n}) = \sum_{\vec q \in \{0,1\}^n}(1-2^{-n})^k = 2^n  (1-2^{-n})^k \\
&\leq 2^n  e^{-k/ 2^{n}} < e^{2n-\frac{k}{2^{n}}}.
\end{align*}
\end{proof}

Any other equilibrium occurs with negligible probability, and has a bounded effect on the loss.
\begin{corollary}
As $k\rightarrow \infty$, the probability that $N$ is the unique equilibrium of $\mj^{B+L}(S_N)$ tends to $1$. 
In particular, $|\loss^{B+L}(n)- \loss^B(n)| \stackrel{k \rightarrow \infty}{\rightarrow} 0$.
\end{corollary}
\rmr{The bound on $k$ is very lax, but this is the simplest proof}

\paragraph{Proxy voting}
From Lemma~\ref{lemma:dictator}, we know that for every set $M\subseteq N$, the most extreme voter $j=1$ gets the votes of all inactive voters with $P_i<0.5$, and in particular is pivotal (w.h.p., as $k$ is large enough). Thus voter~1 is active in any equilibrium, and is in fact a dictator as in the non-strategic scenario.

Finally, since we assume that the median of $h$ is less than $0.5$, $j=1$ is a dictator.  As no other voter in $N$ is pivotal on any issue,  they all become inactive. Thus under the same assumptions of Lemma~\ref{lemma:dictator}:
\begin{corollary}
As $k\rightarrow 0$, the probability that $M=\{1\}$ is the unique equilibrium of $\mj^{S}(S_N)$ tends to $1$. 
In particular, $|\loss^{P+L}(n)- \loss^P(n)| \stackrel{k \rightarrow \infty}{\rightarrow} 0$.
\end{corollary}
\rmr{here we need a much smaller $k$ than in the B+L scenario}

%

\newsubsec{Empirical Evaluation}
We evaluate proxy voting on real data to avoid  two unrealistic assumptions in our theoretical model: that the number of issues $k$ is very large, and that $i$'s votes on all issues are i.i.d.

We examine several data sets from PrefLib~\cite{mattei2013preflib}:\rmr{do not remove citation} The first few datasets are Approval ballots of French presidential 2002 elections over 16 candidates in several regions (ED-26). We treat each candidate is an ``issue'' and each voter can either agree with the issue (approve this candidate) or disagree. $P_i$ is the fraction of issues on which voter~$i$ disagrees with the majority.

%
%

We also considered two datasets of ordinal preferences: sushi preferences (ED-14) and AGH course selection (AD-9). 
The translation to a binary matrix is by checking for each pair of alternatives $(\alpha,\beta)$ whether $\alpha$ is preferred over $\beta$. This leaves us with 45 and 36 binary issues in the sushi and AGH datasets, respectively.\footnote{Note that Hamming distance between agents' positions equals the Kendal-Tau distance between their ordinal preferences.}
  A subset of $k=15$ issues were sampled at each iteration in order to get results that are more robust (we thus get a ``sushi distribution'' and ``AGH distribution'' instead of a single dataset). 

\medskip
We first consider the weight distribution among agents (Fig.~\ref{fig:binary_sushi}). The weight of agents is decreasing in $P_i$, meaning that agents with higher agreement with the majority opinion gets more followers, with the best agents getting a significantly higher weight. This is related to the theoretical result that the best expert get $>0.5$ weight, but is much less extreme. Also there is no weight concentration on the worst agent (this can be explained by the `Anna Karenina principle',\footnote{``Happy families are all alike; every unhappy family is unhappy in its own way''~\cite{tolstoy}.} as each bad agent errs on different issues). In other words, allowing proxies does not result in a dictatorship of the best active agent, but in meritocracy of the better active agents.

This leads us to expect better performance than the theoretical prediction when comparing Proxy voting to the Basic setting. Indeed, Fig.~\ref{fig:binary_sushi} (right) and Fig.~\ref{fig:binary_all} show that in all datasets   $\loss^P(n) < \loss^B(n)$ except for very small samples in the French election datasets. This gap increases quickly with the sample size. 

%
%
%

\begin{figure}[t]
	
	\centering
	\includegraphics[scale=0.32 ]{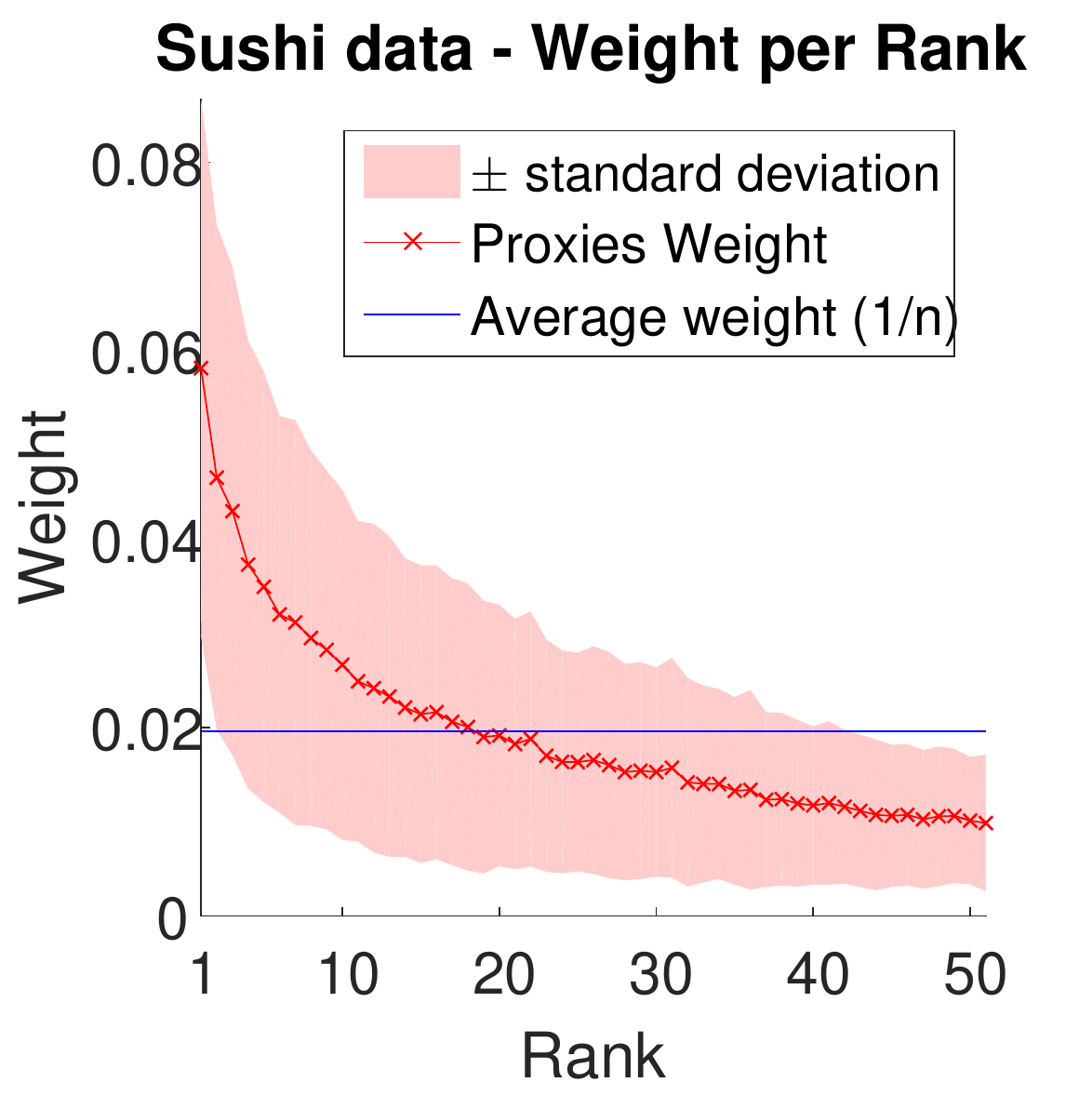}
	~~~
	\includegraphics[scale=0.35 ]{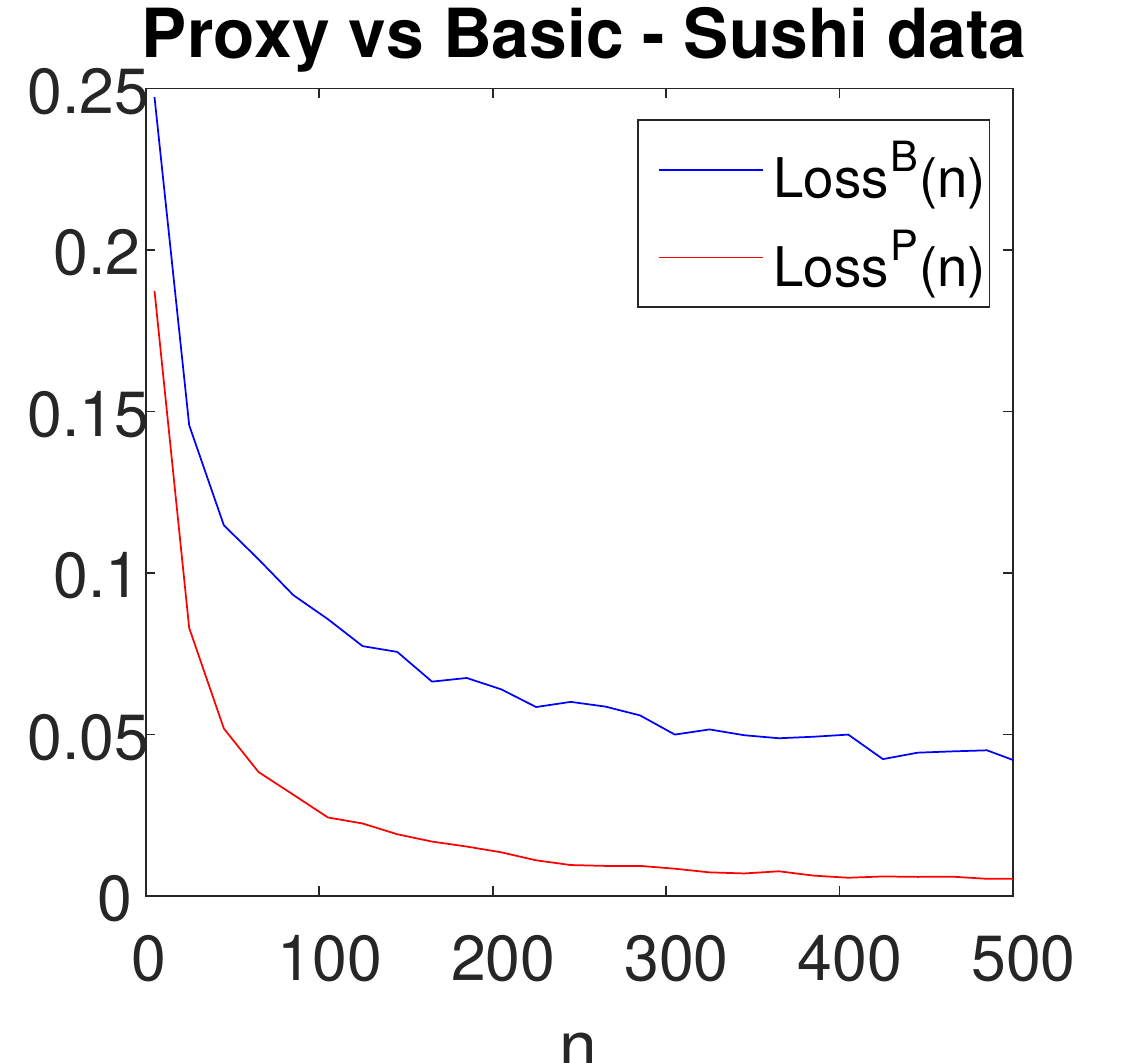}
	\protect\caption{On the left, the average weight of each agent, in increasing order of $P_i$ (best agent on the left).  On the right, the loss with (red) and without (blue) proxies. Results for the other datasets were similar.
		\label{fig:binary_sushi}
		}
\end{figure} 

\begin{figure}[t]
	
	\centering
	\includegraphics[scale=0.45 ]{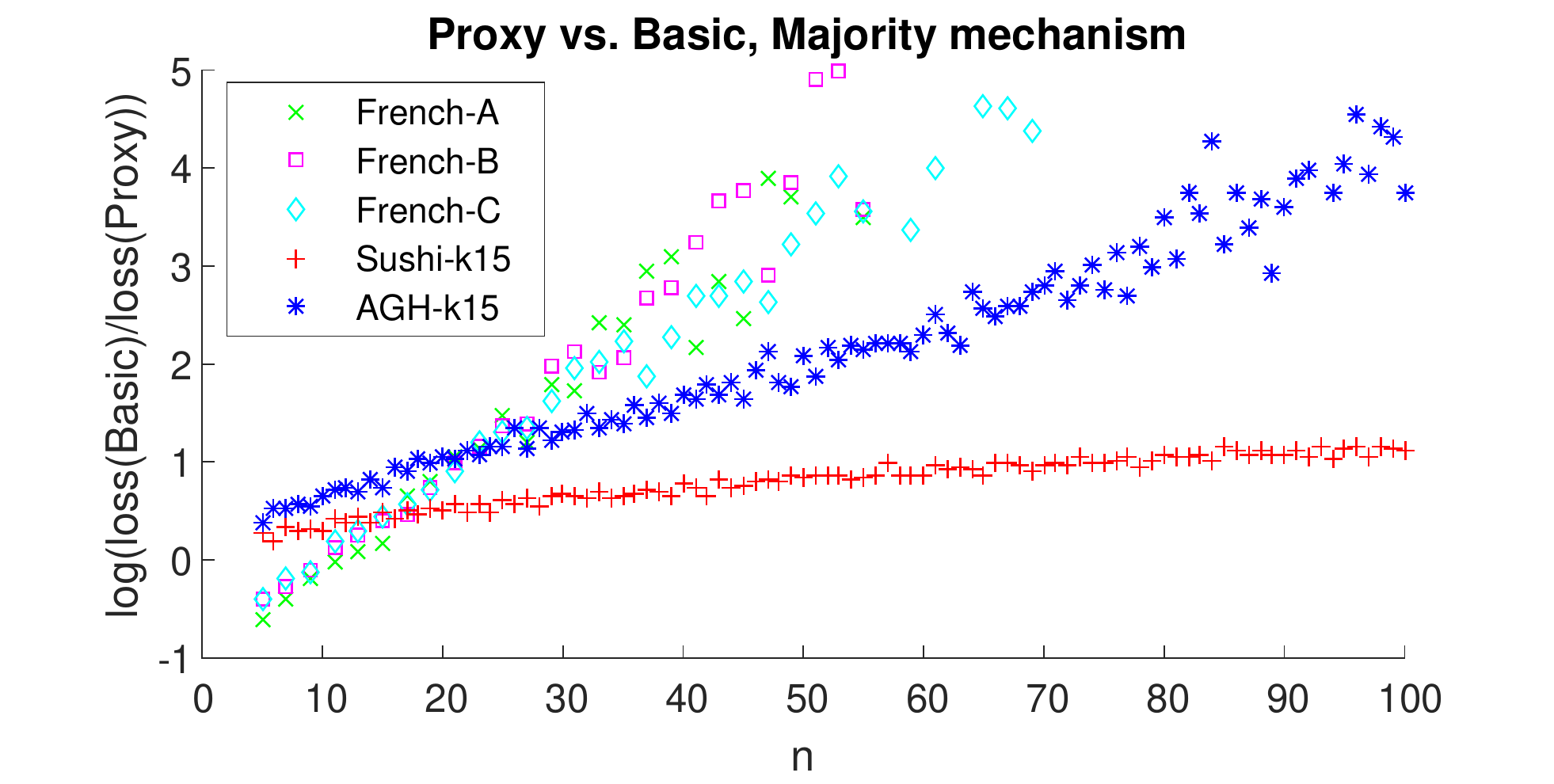}
	\protect\caption{The ratio of $\loss^B(n)$ and $\loss^P(n)$ (in log scale) for all datasets. Each point based on 1000 samples. 
		\label{fig:binary_all}\vspace{-4mm}}
\end{figure} 


\section{Discussion and Related Work}
\label{sec:discussion}
\def\wvs{\vphantom{$2^{2^2}$}}
\begin{table*}[t]
	\begin{center}
		\begin{tabular}{c||c|c|c}
		\hline
			Rule: &   Median  & Mean   &  Majority  ($k\rightarrow \infty$) \\
   				  &			  &		   & with many issues                    \\
			\hline
			\hline
			Proxy better & Yes  & $f$ SP + symmetric   &  No \\
			for any $S_N$& 		& 		$n=2$		   &		\\
			\hline
			$\loss^P<\loss^B$  \wvs    & Always    &  $f$ Uniform & depends on the best  \\
			\cline{1-3}
			$\loss^P\ll \loss^B$  \wvs  &  $f$  symmetric (* most $f$)  &  $f$ Uniform (* most $f$ ) & agent (* real data)\\
			\hline
			\hline
			unique equilibrium &  always   & $f$ Uniform & always \\
			of $\g^{P+L}$\wvs  & 		   & 			 &		   \\
			\hline
			$\loss^{B+L} \geq \loss^B$ \wvs  &  always & always  & always  \\
			\hline
			$\loss^{P+L} \leq \loss^P$ \wvs & always   &  $f$ Uniform (* some SP $f$) & always \\ 
			\hline
		\end{tabular}
	\end{center}
	\protect\caption{A summary of our results. The first three lines show the effect of proxy voting when all agents are active. Results marked by (*) are obtained by simulations. The bottom lines summarize the effect of strategic voting with lazy bias. \label{tab:results} }
\end{table*}

Our results, summarized in Table~\ref{tab:results}, provide a strong support for proxy voting when agents' positions are placed on a line, especially when the Median mechanism is in use. In contrast, when positions are (binary) multi-dimensional, proxy voting might concentrate too much power in the hands of a single proxy, and increase the error. However we also showed that on actual data this rarely happens and analyzed the reasons.  
These findings corroborate our hypothesis that proxy voting can improve representation across several domains. We are looking forward to study the effect of proxy voting in other domains, including common voting functions that use voters' rankings. 

\rmr{connection to learning and boosting?}


Proxy voting, and our model in particular, are tightly related to the \emph{proportional representation problem}, dealing with how to select representatives from a large population. A recent paper by Skowron~\cite{skowron2015we} considers the selection of representatives who then use voting to decide on issues that affect the society. In our case, selection is random as suggested in \cite{mueller1972representative
}, and representatives are weighted proportionally to the number of voters that pick them as proxies, as originally suggested by Tullock~\shortcite{tullock1967proportional}. It is interesting to note that political systems where public representatives are selected at random (``sortition'') have been applied in practice~\cite{dowlen2015political}. Our results suggest that such systems could be improved by weighting the representatives after their selection. Setting the weight \emph{proportionally} to the number of followers seems natural, but it is an open question whether there are even better ways to set these weights.
%

Closest to out work is a model by Green-Armytage~\shortcite{green2015direct}, where voters select proxies and use the Median rule to decide on each of several continuous issues. Decisions are evaluated based on their square distance from the ``optimal'' one. However even if the entire population votes, the outcome may be suboptimal, as Green-Armytage assumes people perceive their own position (as well as others' position) with some error. He then focuses on how various options for delegating one's vote may contribute to reducing her \emph{expressive loss}, i.e. the distance from her true opinion to her ballot.
In contrast, expressive losses do not play a roll in our model, where the sources of inaccuracy are small samples and/or strategic behavior. 

Alger~\shortcite{alger2006voting} considers a model with a fixed set of political representatives on an interval (as in our model), but focuses mainly on the ideological considerations of the voters and the political implications rather than on mathematical analysis.  Our very positive results on the use of proxies in the Median mechanism support Alger's conclusions, albeit under a somewhat different model of voters incentives. Alger also points out that proxy voting significantly reduces the amount of communication involved in collecting ballots on many issues. 

Other models allow chains of voters who use each other as proxies~\cite{green2005direct,degrave2014resolving}, or social influence that  effectively increases the weight of some voters~\cite{alon2015robust}.

Indeed, we believe that a realistic model of proxy voting would have to take into account such topological and social factors in addition to statistics and incentives. E.g., \cite{alon2015robust} shows the benefits of a bounded degree, which in our model may allow a way to bound excessive weights.
Social networks may also be a good way to capture correlations in voters' preferences~\cite{procaccia2015ranked}, and can thus be used to extend our results beyond independent voters.
%

\newpar{Strategic behavior} 
We showed that most of our results hold when participation is strategic. What if voters  (either active or inactive) could mis-report their position? Note that inactive voters have no reason to lie under the Median and the Majority mechanisms, due to standard strategyproofness properties.  However active agents may be able to affect the outcome by changing the partition of followers.  
\rmr{Another variation is to consider variable costs for proxy voting, which may depend on the location of the voter, her distance from her proxy etc.}
We can also consider more nuanced strategic behavior, for example where an agent also cares about her number of followers regardless of the outcome. More generally, strategic considerations under proxy voting combine challenges from strategic voting with those of strategic candidacy~\cite{hotelling29,dutta2001strategic}, and would require a careful review of the assumptions of each model. 


Other open questions include the effect of proxy voting on \emph{diversity}, \emph{fairness}, and \emph{participation}. It is argued that diverse representatives often reach better outcomes~\cite{marcolino2013multi}, and fairness attracts much attention in the analysis of voting and other multiagent systems~\cite{bogomolnaia2005collective,walsh2007representing,dickerson2014computational}. The effect on participation and engagement may also be quite involved, since allowing voters to use a proxy may increase the participation level of some who would otherwise not be represented, but on the other hand may lower the incentive to vote actively, thereby reducing overall engagement of the society.

Finally, the future of proxy voting depends on the development and penetration of novel online voting tools and social apps, such as those mentioned in the Introduction. We hope that sharing of data and insights will promote research on the topic, and set new challenges for mechanism design.
 
\rmr{In our paper at least lazy bias indeed leads to reduced participation under proxy voting, but without increasing the error. In more complex models this may no longer hold.}

%


%
%


\input{Proxy_voting_Arxiv.bbl}
%
\end{document}

%% file: defs.tex
\newtheorem{theorem}{Theorem}
\newtheorem{corollary}[theorem]{Corollary}
\newtheorem{lemma}[theorem]{Lemma}

\newtheorem{proposition}[theorem]{Proposition}

\newtheorem{conjecture}[theorem]{Conjecture}

\newcommand{\sproof}{\noindent{\bf Proof.}\hspace*{1em}}

\def\eps{{\varepsilon}}

\def\literalqed{{\ \nolinebreak\hfill\mbox{\quad}}}
\long\def\symbolfootnote[#1]#2{\begingroup%
\def\thefootnote{\fnsymbol{footnote}}\footnote[#1]{#2}\endgroup} 

    \makeatletter
    \renewcommand\part{%
      \if@openright
        \cleardoublepage
      \else
        \clearpage
      \fi
      \thispagestyle{empty}%
      \if@twocolumn
        \onecolumn
        \@tempswatrue
      \else
        \@tempswafalse
      \fi
      \null\vfil
      \secdef\@part\@spart}
    \makeatother

\def\calX{{\cal X}}

\def\loss{\ell}

\def\eps{\epsilon}
\DeclareMathOperator{\argmin}{\text{argmin}}
\DeclareMathOperator{\argmax}{\text{argmax}}

\newcommand{\labeq}[2]{\begin{equation}\label{eq:#1}#2\end{equation}}

\newcommand{\ind}[1]{\llbracket #1 \rrbracket}

\newcommand{\floor}[1]{\left\lfloor #1 \right\rfloor}

\def\({\left(}
\def\){\right)}

\renewcommand{\ind}[1]{\left\llbracket #1 \right\rrbracket}
\newcommand{\xqed}{\mbox{\raggedright $\Diamond$}}
\def\xqedhere{\hfill\xqed}

\def\newpar{\vspace{-3mm}\paragraph}

\renewcommand{\vec}{\mathbf}
\def\cal{\mathcal}
\def\g{\mathbf{g}}
\def\md{\mathbf{md}}
\def\mn{\mathbf{mn}}
\def\mj{\mathbf{mj}}
\def\loss{\mathcal{L}}

%% file: line_fig.tex
%
%
%
%
%

\begin{center}

\begin{tikzpicture}[scale=0.8,transform shape]
\def\hh{1.8}
\def\rr{0.08}
\draw[thick,|-|] (0,\hh) -- (10,\hh);
\node [above] at (0,\hh+0.1) {$0$};
\node [above] at (10,\hh+0.1) {$10$};

\draw[fill] (1,\hh) circle [radius=\rr];
\draw[fill] (3,\hh) circle [radius=\rr];
\draw[fill] (6,\hh) circle [radius=\rr];
\draw[fill] (7,\hh) circle [radius=\rr];
\node [below] at (1,\hh) {$s_1=1$};
\node [below] at (3,\hh) {$s_2=3$};
\node [below] at (5.8,\hh) {$s_3=6$};
\node [below] at (7.5,\hh) {$s_4=7$};

\draw[->] (2,\hh+0.3) -- (3-0.2,\hh+0.1); 
\node [above] at (1.5,\hh+0.1) {\small{$\md(S_N)$}};
\draw[->] (17/4-0.5,\hh+0.3) -- (17/4,\hh+0.03);
\node [above] at (17/4-0.5,\hh+0.3) {\small{$\mn(S_N)=4.25$}};


\draw[thick,|-|] (0,0) -- (10,0);
\draw[fill] (1,0) circle [radius=\rr];
\draw[fill] (3,0) circle [radius=\rr];
\draw[fill] (6,0) circle [radius=\rr];
\draw[fill] (7,0) circle [radius=\rr];
\draw[|<->|] (0,-0.5) -- (2,-0.5);
\node [below] at (1,-0) {$w_1=2$};
\draw[|<->|] (2.0,-0.2) -- (4.5,-0.2);
\node [below] at (3,-0.2) {$w_2=2.5$};
\draw[|<->|] (4.5,-0.5) -- (6.5,-0.5);
\node [below] at (5.8,-0) {$w_3=2$};
\draw[|<->|] (6.5,-0.2) -- (10,-0.2);
\node [below] at (7.5,-0.2) {$w_4=3.5$};

\draw[->] (6.5,0.3) -- (6+0.1,0.1); 
\node [above] at (6.5,0.3) {\small{$\md(S_N,\vec w_N)$}};
\draw[->] (4.6-1,0.3) -- (4.6,0.03);
\node [above] at (4.6-1,0.3) {\small{$\mn(S_N,\vec w_N)=4.6$}};


\draw[thick,|-|] (0,0-\hh) -- (10,0-\hh);
\draw[fill] (1,0-\hh) circle [radius=\rr];
\draw[fill=white] (3,0-\hh) circle [radius=\rr];
\draw[fill] (6,0-\hh) circle [radius=\rr];
\draw[fill] (7,0-\hh) circle [radius=\rr];
\draw[|<->|] (0,-0.2-\hh) -- (3.5,-0.2-\hh);
\node [below] at (1,-0.2-\hh) {$w_1=3.5$};
\draw[|<->|] (3.5,-0.5-\hh) -- (6.5,-0.5-\hh);
\node [below] at (5.8,-\hh) {$w_3=3$};
\draw[|<->|] (6.5,-0.2-\hh) -- (10,-0.2-\hh);
\node [below] at (7.5,-0.2-\hh) {$w_4=3.5$};

\draw[->] (6.5,0.3-\hh) -- (6+0.1,0.1-\hh); 
\node [above] at (6.5,0.3-\hh) {\small{$\md(S_M,\vec w_M)$}};
\draw[->] (4.6-1,0.3-\hh) -- (4.6,0.03-\hh);
\node [above] at (4.6-1,0.3-\hh) {\small{$\mn(S_M,\vec w_M)=4.6$}};


\draw[dotted] (5,\hh+0.2) -- (5,-0.2-\hh);

\end{tikzpicture}
\end{center}
%

%% file: Proxy_voting_Arxiv.bbl
\begin{thebibliography}{10}

\bibitem{alger2006voting}
Dan Alger.
\newblock Voting by proxy.
\newblock {\em Public Choice}, 126(1-2):1--26, 2006.

\bibitem{alon2015robust}
Noga Alon, Michal Feldman, Omer Lev, and Moshe Tennenholtz.
\newblock How robust is the wisdom of the crowds?
\newblock In {\em IJCAI'15}, 2015.

\bibitem{arnold1992first}
Barry~C Arnold, Narayanaswamy Balakrishnan, and Haikady~Navada Nagaraja.
\newblock {\em A first course in order statistics}, volume~54.
\newblock Siam, 1992.

\bibitem{bogomolnaia2005collective}
Anna Bogomolnaia, Herv{\'e} Moulin, and Richard Stong.
\newblock Collective choice under dichotomous preferences.
\newblock {\em Journal of Economic Theory}, 122(2):165--184, 2005.

\bibitem{christian2005voting}
L~Christian~Schaupp and Lemuria Carter.
\newblock E-voting: from apathy to adoption.
\newblock {\em Journal of Enterprise Information Management}, 18(5):586--601,
  2005.

\bibitem{degrave2014resolving}
Jonas Degrave.
\newblock Resolving multi-proxy transitive vote delegation.
\newblock {\em arXiv preprint arXiv:1412.4039}, 2014.

\bibitem{dickerson2014computational}
John~P Dickerson, Jonathan~R Goldman, Jeremy Karp, Ariel~D Procaccia, and
  Tuomas Sandholm.
\newblock The computational rise and fall of fairness.
\newblock In {\em AAAI'14}, pages 1405--1411, 2014.

\bibitem{dowlen2015political}
Oliver Dowlen.
\newblock {\em The political potential of sortition: A study of the random
  selection of citizens for public office}, volume~4.
\newblock Andrews UK Limited, 2015.

\bibitem{dutta2001strategic}
Bhaskar Dutta, Matthew~O Jackson, and Michel Le~Breton.
\newblock Strategic candidacy and voting procedures.
\newblock {\em Econometrica}, 69(4):1013--1037, 2001.

\bibitem{elkind2015equilibria}
Edith Elkind, Evangelos Markakis, Svetlana Obraztsova, and Piotr Skowron.
\newblock Equilibria of plurality voting: Lazy and truth-biased voters.
\newblock In {\em SAGT'15}, pages 110--122. Springer, 2015.

\bibitem{green2005direct}
James Green-Armytage.
\newblock Direct democracy by delegable proxy.
\newblock {\em DOI= http://fc. antioch. edu/\~{} james\_greenarmytage/vm/proxy.
  htm}, 2005.

\bibitem{green2015direct}
James Green-Armytage.
\newblock Direct voting and proxy voting.
\newblock {\em Constitutional Political Economy}, 26(2):190--220, 2015.

\bibitem{grofman1983thirteen}
Bernard Grofman, Guillermo Owen, and Scott~L Feld.
\newblock Thirteen theorems in search of the truth.
\newblock {\em Theory and Decision}, 15(3):261--278, 1983.

\bibitem{hotelling29}
Harold Hotelling.
\newblock Stability in competition.
\newblock {\em The Economic Journal}, 39(153):41--57, 1929.

\bibitem{jonsson2011user}
Anna~Maria J{\"o}nsson and Henrik {\"O}rnebring.
\newblock User-generated content and the news: empowerment of citizens or
  interactive illusion?
\newblock {\em Journalism Practice}, 5(2):127--144, 2011.

\bibitem{marcolino2013multi}
Leandro~Soriano Marcolino, Albert~Xin Jiang, and Milind Tambe.
\newblock Multi-agent team formation: diversity beats strength?
\newblock In {\em IJCAI}, 2013.

\bibitem{mattei2013preflib}
Nicholas Mattei and Toby Walsh.
\newblock Preflib: A library for preferences http://www. preflib. org.
\newblock In {\em ADT'13}, pages 259--270, 2013.

\bibitem{miller1969program}
James~C Miller~III.
\newblock A program for direct and proxy voting in the legislative process.
\newblock {\em Public choice}, 7(1):107--113, 1969.

\bibitem{moulin1980strategy}
Herv{\'e} Moulin.
\newblock On strategy-proofness and single peakedness.
\newblock {\em Public Choice}, 35(4):437--455, 1980.

\bibitem{mueller1972representative}
Dennis~C Mueller, Robert~D Tollison, and Thomas~D Willett.
\newblock Representative democracy via random selection.
\newblock {\em Public Choice}, 12(1):57--68, 1972.

\bibitem{petrik2009participation}
Klaus Petrik.
\newblock Participation and e-democracy how to utilize web 2.0 for policy
  decision-making.
\newblock In {\em DGO'09}, pages 254--263, 2009.

\bibitem{procaccia2015ranked}
Ariel~D Procaccia, Nisarg Shah, and Eric Sodomka.
\newblock Ranked voting on social networks.
\newblock AAAI, 2015.

\bibitem{riddick1985riddick}
Floyd Riddick and Miriam Butcher.
\newblock {\em Riddick's Rules of Procedure}.
\newblock Lanham, MD: Madison Books, 1985.

\bibitem{skowron2015we}
Piotr Skowron.
\newblock What do we elect committees for? a voting committee model for
  multi-winner rules.
\newblock In {\em Proceedings of the 24th International Joint Conference on
  Artificial Intelligence (IJCAI-2015)}, pages 1141--1148, 2015.

\bibitem{stigler1973studies}
Stephen~M Stigler.
\newblock Studies in the history of probability and statistics. xxxii laplace,
  fisher, and the discovery of the concept of sufficiency.
\newblock {\em Biometrika}, 60(3):439--445, 1973.

\bibitem{tolstoy}
Leo Tolstoy.
\newblock {\em Anna Karenina}.
\newblock 1877.

\bibitem{tullock1967proportional}
Gordon Tullock.
\newblock Proportional representation.
\newblock {\em Toward a mathematics of politics}, pages 144--157, 1967.

\bibitem{wackerly2007mathematical}
Dennis Wackerly, William Mendenhall, and Richard~L Scheaffer.
\newblock {\em Mathematical statistics with applications}.
\newblock Nelson Education, 2007.

\bibitem{walsh2007representing}
Toby Walsh.
\newblock Representing and reasoning with preferences.
\newblock {\em AI Magazine}, 28(4):59, 2007.

\end{thebibliography}
